\documentclass[reqno, eucal]{amsart}
\usepackage{amsmath,amssymb,amsthm,amscd}
\usepackage[mathscr]{eucal}
\usepackage{epic,eepic}
\usepackage{psfrag}

\numberwithin{equation}{section}

\newtheorem{theorem}{Theorem}
\newtheorem{lemma}[theorem]{Lemma}

\newtheorem{proposition}[theorem]{Proposition}

\theoremstyle{definition}

\newtheorem{remark}[theorem]{Remark}

\newcommand{\ot}{\otimes}

\newcommand{\Z}{{\mathbb Z}}

\newcommand{\C}{{\mathbb C}}
\newcommand{\Q}{{\mathbb Q}}
\newcommand{\Rm}{{\mathscr R}}
\newcommand{\ket}[1]{|#1\rangle}

\begin{document}
\title[Tetrahedron equation and quantum $R$ matrices]
{Tetrahedron equation and quantum $\boldsymbol{R}$ matrices\\
for $\boldsymbol{q}$-oscillator representations of \\
$\boldsymbol{U_q(A^{(2)}_{2n})}$, 
$\boldsymbol{U_q(C^{(1)}_{n})}$ and 
$\boldsymbol{U_q(D^{(2)}_{n+1})}$}

\author{Atsuo Kuniba}
\email{atsuo@gokutan.c.u-tokyo.ac.jp}
\address{Institute of Physics, Graduate School of Arts and Sciences,
University of Tokyo, Komaba, Tokyo 153-8902, Japan}

\author{Masato Okado}
\email{okado@sci.osaka-cu.ac.jp}
\address{Department of Mathematics, Osaka City University, 
3-3-138, Sugimoto, Sumiyoshi-ku, Osaka, 558-8585, Japan}

\maketitle

\vspace{1cm}
\begin{center}{\bf Abstract}\end{center}
\vspace{0.2cm}

The intertwiner of the quantized coordinate ring $A_q(sl_3)$
is known to yield a solution to the tetrahedron equation.
By evaluating their $n$-fold composition with
special boundary vectors 
we generate series of solutions to the Yang-Baxter equation.
Finding their origin in conventional
quantum group theory is a clue to the link between
two and three dimensional integrable systems.
We identify them with the quantum $R$ matrices
associated with the 
$q$-oscillator representations of $U_q(A^{(2)}_{2n})$, 
$U_q(C^{(1)}_n)$ and $U_q(D^{(2)}_{n+1})$.

\section{Introduction}\label{sec:1}

The tetrahedron equation \cite{Zam80} is a generalization of 
the Yang-Baxter equation \cite{Bax} and 
serves as a key to the quantum integrability in three dimension (3d).
Among its many formulations
the homogeneous version of vertex type has the form
\begin{align*}
\Rm_{1,2,4}\Rm_{1,3,5}\Rm_{2,3,6}\Rm_{4,5,6}=
\Rm_{4,5,6}\Rm_{2,3,6}\Rm_{1,3,5}\Rm_{1,2,4},
\end{align*}
where $\Rm$ is a linear operator 
on the tensor cube of some vector space $F$.
The equality holds in $\mathrm{End}(F^{\otimes 6})$ 
where the indices indicate the 
components on which each $\Rm$ acts nontrivially.
We call a solution to the tetrahedron equation a 3d $R$.

In the tetrahedron equation one sees that 
if the spaces $4,5$ and $6$  are evaluated away
appropriately, it reduces to the Yang-Baxter equation:
\begin{align*}
\Rm_{1,2}\Rm_{1,3}\Rm_{2,3}=
\Rm_{2,3}\Rm_{1,3}\Rm_{1,2}.
\end{align*}
By now, algebraic background of the Yang-Baxter equation
has been well understood by the representation theory of quantum groups and
their generalizations. 
Thus the following problem arises.
Given a 3d $R$, find a prescription to reduce it to a solution of the Yang-Baxter
equation and clarify its context in the framework of quantum group theory.   
It is a clue to the connection between integrability in two and three dimensions.

In this paper we present the solution of the problem
for the distinguished example of 3d $R$ $\Rm$ \cite{KV, BS}.
It acts on the tensor cube of the $q$-oscillator Fock space $F$ and 
possesses several remarkable features related to the quantized coordinate ring 
$A_q(sl_3)$, the PBW bases of the
nilpotent subalgebra of $U_q(sl_3)$ and so on.
See Section \ref{ss:3dr} and Appendix \ref{app:3dR} 
for more accounts on the $\Rm$.

Our prescription for the reduction is parallel with the 
earlier work concerning 3d $L$ operator \cite{KS}.
Namely we take matrix elements of the tetrahedron equation 
by using certain vectors in the 4, 5 and 6 th components in $F^{\otimes 6}$.
These vectors contain a spectral parameter $z$  
and serve as special boundary conditions in the context of the 
3d lattice model associated with the $\Rm$.
In fact the tetrahedron equation itself admits a 
straightforward extension to the $n$-site situation 
(see (\ref{TEn})) for which the reduction works equally.
In this way the single 3d $R$ $\Rm$ yields infinite series of solutions of the 
Yang-Baxter equation labeled by $n$ and the boundary vectors. 
Up to an overall scalar they are rational functions of $q$ and 
the (multiplicative) spectral parameter $z$, leading to 
integrable 2d vertex models having local states in $F^{\otimes n}$.

Our main result is Theorem \ref{th:main}, which identifies 
these solutions of the Yang-Baxter equation
with the quantum $R$ matrices for the 
$q$-oscillator representations of 
$U_q(D^{(2)}_{n+1})$, $U_q(A^{(2)}_{2n})$ and $U_q(C^{(1)}_n)$
on $F^{\otimes n}$ depending on the choice of the boundary vectors.
Namely the solutions coincide with the 
intertwiner of the tensor products up to an overall scalar.
Actually Theorem \ref{th:main} has also 
guided us to introduce the $q$-oscillator representation themselves.
For type $A_n^{(1)}$ or $C_n$ they were introduced in
\cite{Ha} using a $q$-analogue
of the Weyl algebra.
Apart from complementing the latter to $C^{(1)}_n$,  
the $q$-oscillator representations for type 
$D^{(2)}_{n+1}$ and $A^{(2)}_{2n}$ in this paper
containing $B_n$ as a classical part seem to be new.
An intriguing feature of them is that 
the quantum parameter $q$ cannot be specialized to be 1.
A similar singularity at $q=1$ has been known for 
the unitary representations of 
non-compact real forms of $U_q(sl_2)$ \cite{MMNNSU,P}.
However our case has another distinctive aspect that 
the action of some weight generators $k_j$ acquire the factor $i=\sqrt{-1}$
besides a power of $q$.

We have done the task of determining 
the spectral decomposition of the associated new quantum 
$R$ matrices which consists of 
infinitely many irreducible components.
It provides the information 
complementary with the explicit formula (\ref{sabij}).

This work is motivated by several preceding results.
A relation between the tetrahedron equation and quantum $R$ matrices goes back, for
example, to \cite{S:1997,KV:2000}.
In \cite{BS}, the reduction was made for 
the same $\Rm$ by taking the {\em trace} and 
the consequent solution to the Yang-Baxter equation 
was announced to be the direct sum of the 
quantum $R$ matrices for the symmetric tensor representations of 
$U_q(A^{(1)}_{n-1})$.
We have summarized it in Appendix \ref{app:A}
for comparison.
A further result on the trace reduction is available
for $n=2$ \cite{BMS2}.
In \cite{KS}, the reduction based on the 
same boundary vectors as this paper was studied for the
$n$-product of the 3d $L$ operator \cite{BS}.
The result was identified with the quantum $R$ matrices for the 
spin representations of 
$U_q(B^{(1)}_{n}), U_q(D^{(1)}_{n})$ \cite{O}
and $U_q(D^{(2)}_{n+1})$.
See Remark \ref{re:dyn} and \cite[Remark 7.2]{KS}
for the comparison of these quantum affine algebras and 
those captured in this paper.
A notable fact is that the boundary vectors specify  
the end shape of the Dynkin diagram 
of the relevant affine Lie algebras. 
In our previous paper \cite{KO3}, 
the reduction using the same boundary vectors 
was applied to the {\em single} $\Rm$ 
and the result was identified with 
the quantum $R$ matrices for $q$-oscillator representations of 
the {\em rank one} quantum affine algebras
$U_q(A^{(1)}_1)$ and $U_q(A^{(2)}_2)$.
The present paper contains these results as the $n=1$ case 
by regarding $U_q(D^{(2)}_2)$ and $U_q(C^{(1)}_1)$ 
as $U_q(A^{(1)}_1)$ appropriately.
We note that a more general problem of studying the 
{\em mixture} of 3d $R$ and $L$ operators has been formulated 
and the simplest case has been worked out in \cite[Section 5]{KO3}.

The outline of the paper is as follows.
In Section \ref{sec:3dR} we recall the prescription \cite{KS,KO3}
to generate solutions to the Yang-Baxter equation 
from a solution to the tetrahedron equation 
using boundary vectors.
We then apply it to the 3d $R$ (\ref{Rex}) 
acting on the tensor cube $F^{\otimes 3}$ of the Fock space of
$q$-oscillators.
There are two boundary vectors leading to
the four families of solutions $S^{s,t}(z) \;(s,t=1,2)$ 
of the Yang-Baxter equation (\ref{sact})--(\ref{sabij}). 
They correspond to vertex models on planar square lattice whose local states 
range over $F^{\otimes n}$.

In Section \ref{sec:R} we introduce the $q$-oscillator representations of 
the Drinfeld-Jimbo quantum affine algebras 
$U_q(D^{(2)}_{n+1})$, $U_q(A^{(2)}_{2n})$ and $U_q(C^{(1)}_n)$.
Their tensor product decomposes into a direct sum of infinitely many 
irreducible submodules with respect to the classical part 
$U_q(B_n)$ or $U_q(C_n)$.
The spectral decomposition of the $R$ matrices is
done in Section \ref{ss:sv}  and  \ref{ss:sd}, 
although this part is not used in the rest. 

In Section \ref{sec:main}
we give our main result Theorem \ref{th:main}.
It identifies the solutions $S^{s,t}(z)$ of the Yang-Baxter equation 
with the quantum $R$ matrices for $q$-oscillator representations.
Depending on the choice of the boundary vectors,   
$U_q(D^{(2)}_{n+1})$, $U_q(A^{(2)}_{2n})$ and $U_q(C^{(1)}_n)$ cases
are covered.
The correspondence 
between the boundary vectors and the Dynkin diagrams of the 
relevant affine Lie algebras in Remark \ref{re:dyn}
is parallel but not identical to the earlier 
observation in \cite[Remark 7.2]{KS} concerning 3d $L$ operators.
Our proof of Theorem \ref{th:main} is done by using the  
characterization of the quantum $R$ matrices without recourse to 
their explicit forms.
It implies that the
commutativity with the $q$-oscillator representation of $U_q$ 
is embedded into intertwining relations of the quantized
coordinate ring $A_q$ through the evaluation by 
boundary vectors.

Appendix \ref{app:3dR} contains a brief guide to 
the 3d $R$ $\Rm$ from the representation theory of 
the quantized coordinate ring $A_q(sl_3)$ \cite{KV}. 
All the lemmas necessary for the proof of Theorem \ref{th:main}
are prepared.

Appendix \ref{app:A} is an exposition of the $U_q(A^{(1)}_{n-1})$ case
in the setting of this paper.
It is relevant to the trace reduction of the tetrahedron equation \cite{BS}.

Throughout the paper we assume that $q$ is generic and 
use the following notations:
\begin{align*}
&(z;q)_m = \prod_{k=1}^m(1-z q^{k-1}),\;\;
(q)_m = (q; q)_m,\;\;
\binom{m}{k}_{\!\!q}= \frac{(q)_m}{(q)_k(q)_{m-k}},\\
&[m]=[m]_q = \frac{q^m-q^{-m}}{q-q^{-1}}, \;\;
[m]_q! = \prod_{k=1}^m [k]_q,\;\;
{m\brack k} = \frac{[m]!}{[k]![m-k]!},
\end{align*}
where the both $q$-binomials are to be understood as zero
unless $0 \le k \le m$.

\section{Solutions of the Yang-Baxter equation from 3d $R$}\label{sec:3dR}
\subsection{General scheme}\label{ss:gs}
Let  $F$ be a vector space and 
$\Rm \in \mathrm{End}(F^{\otimes 3})$.
Consider the tetrahedron equation:
\begin{align}\label{TE}
\Rm_{1,2,4}\Rm_{1,3,5}\Rm_{2,3,6}\Rm_{4,5,6}=
\Rm_{4,5,6}\Rm_{2,3,6}\Rm_{1,3,5}\Rm_{1,2,4},
\end{align}
which is an equality in $\mathrm{End}(F^{\otimes 6})$.
Here $\Rm_{i,j,k}$ acts as $\Rm$ on the 
$i,j,k$ th components from the left in the 
tensor product $F^{\otimes 6}$.

Let us recall the prescription which produces  
an infinite family of solutions to the Yang-Baxter equation 
from a solution to the tetrahedron equation based on 
special boundary vectors \cite{KS}.
First we regard (\ref{TE}) as a one-site relation,
and extend it to the $n$-site version rather straightforwardly.
Let $\overset{\alpha_i}{F},
\overset{\beta_i}{F},
\overset{\gamma_i}{F}$ be the copies of $F$,
where $\alpha_i, \beta_i$ and $\gamma_i\,(i=1,\ldots, n)$
are just labels and not parameters.
Renaming the spaces $1,2,3$ by them,
we have 
$\Rm_{\alpha_i, \beta_i, 4}
\Rm_{\alpha_i, \gamma_i, 5}
\Rm_{\beta_i, \gamma_i, 6}\Rm_{4,5,6}=
\Rm_{4,5,6}
\Rm_{\beta_i, \gamma_i, 6}
\Rm_{\alpha_i, \gamma_i, 5}
\Rm_{\alpha_i, \beta_i, 4}$
for each $i$.
Thus for any $i$ one can carry $\Rm_{4,5,6}$ through  
$\Rm_{\alpha_i, \beta_i, 4}
\Rm_{\alpha_i, \gamma_i, 5}
\Rm_{\beta_i, \gamma_i, 6}$ to the left 
converting it to the reverse order product
$\Rm_{\beta_i, \gamma_i, 6}
\Rm_{\alpha_i, \gamma_i, 5}
\Rm_{\alpha_i, \beta_i, 4}$.
Repeating this $n$ times leads to the relation
\begin{equation}\label{TEn}
\begin{split}
&\bigl(\Rm_{\alpha_1, \beta_1, 4}
\Rm_{\alpha_1, \gamma_1, 5}
\Rm_{\beta_1, \gamma_1, 6}\bigr)
\cdots 
\bigl(\Rm_{\alpha_n, \beta_n, 4}
\Rm_{\alpha_n, \gamma_n, 5}
\Rm_{\beta_n, \gamma_n, 6}\bigr)\Rm_{4,5,6}\\
&=
\Rm_{4,5,6}
\bigl(\Rm_{\beta_1, \gamma_1, 6}
\Rm_{\alpha_1, \gamma_1, 5}
\Rm_{\alpha_1, \beta_1, 4}
\bigr)
\cdots 
\bigl(\Rm_{\beta_n, \gamma_n, 6}
\Rm_{\alpha_n, \gamma_n, 5}
\Rm_{\alpha_n, \beta_n, 4}
\bigr).
\end{split}
\end{equation}
This is an equality in 
$\mathrm{End}(\overset{\boldsymbol \alpha}{F}\otimes 
\overset{\boldsymbol\beta}{F}\otimes 
\overset{\boldsymbol\gamma}{F}\otimes 
\overset{4}{F}\otimes 
\overset{5}{F}\otimes 
\overset{6}{F})$,
where 
${\boldsymbol\alpha}=(\alpha_1,\ldots, \alpha_n)$
is the array of labels and 
$\overset{\boldsymbol \alpha}{F}= 
\overset{\alpha_1}{F}\otimes \cdots \otimes 
\overset{\alpha_n}{F}\,(= F^{\otimes n})$.
The notations  
$\overset{\boldsymbol\beta}{F}$ and $\overset{\boldsymbol\gamma}{F}$
should be understood similarly.
The argument so far is just a 3d analogue of the simple 
fact in 2d that a single $RLL = LLR$ relation for a local $L$ operator 
implies a similar relation for the $n$-site monodromy matrix
in the quantum inverse scattering method.

Now we turn to the special boundary vectors.
Suppose we have a vector
$|\chi_s(x)\rangle \in F$
depending on a variable $x$ such that 
its tensor product
\begin{equation}\label{cv}
|\chi_s(x,y)\rangle
=|\chi_s(x)\rangle \otimes |\chi_s(xy)\rangle
\otimes |\chi_s(y)\rangle
\in F\otimes F\otimes F
\end{equation}
satisfies the relation
\begin{equation}\label{RX}
\Rm |\chi_s(x,y)\rangle   = |\chi_s(x,y)\rangle.
\end{equation}
The index $s$ is put to distinguish possibly more than one such vectors.
Suppose there exist vectors in the dual space
\begin{equation*}
\langle \chi_s(x,y)|
=\langle \chi_s(x)| \otimes
\langle \chi_s(xy)| \otimes
\langle \chi_s(y)|
\in F^*\otimes F^*\otimes F^*
\end{equation*}
having the similar property
\begin{equation}\label{XR}
\langle \chi_s(x,y) |\Rm =\langle \chi_s(x,y) | .
\end{equation}
Then evaluating (\ref{TEn}) between 
$\langle \chi_s(x,y) |$ and $|\chi_t(1,1)\rangle$\footnote{
This could be chosen as $|\chi_t(x',y')\rangle$ in general. 
However in all the examples studied in this paper, such a freedom is 
absorbed into $(x,y)$.},
one encounters the object
\begin{align}\label{sdef}
S_{\boldsymbol{\alpha, \beta}}(z)=\varrho^{s,t}(z)\langle \chi_s(z)|
\Rm_{\alpha_1, \beta_1, 3}
\Rm_{\alpha_2, \beta_2, 3}\cdots
\Rm_{\alpha_n, \beta_n, 3}
|\chi_t(1)\rangle
\in \mathrm{End}
(\overset{\boldsymbol\alpha}{F}\otimes \overset{\boldsymbol\beta}{F}),
\end{align}
where the scalar $\varrho^{s,t}(z)$ is inserted to control the normalization.
The composition of $\Rm$ and 
matrix elements are taken with respect to the space signified by $3$.
Plainly one may write it as
$S(z) \in \mathrm{End}(F^{\otimes n} \otimes F^{\otimes n})$
removing the dummy labels.
Remember that $S(z)$ of course depends on $s$ and $t$ 
although they have been temporarily suppressed in the notation.
It follows from (\ref{TEn}), (\ref{RX}) and (\ref{XR}) that 
$S(z)$ satisfies the Yang-Baxter equation:
\begin{align}\label{sybe}
S_{\boldsymbol{\alpha, \beta}}(x)
S_{\boldsymbol{\alpha, \gamma}}(xy)
S_{\boldsymbol{\beta,\gamma}}(y)
=
S_{\boldsymbol{\beta,\gamma}}(y)
S_{\boldsymbol{\alpha, \gamma}}(xy)
S_{\boldsymbol{\alpha, \beta}}(x)
\in \mathrm{End}(
\overset{\boldsymbol\alpha}{F}\otimes
\overset{\boldsymbol\beta}{F}\otimes
\overset{\boldsymbol\gamma}{F}).
\end{align}
This fact holds for each choice of $(s,t)$.
See \cite[Section 5]{KO3} for a further 
generalization of the procedure to deduce  
solutions to the Yang-Baxter equation by
mixing more than one kind of solutions 
to the tetrahedron equation.

\subsection{\mathversion{bold}3d $R$ and boundary vectors}\label{ss:3dr}
Let us proceed to a concrete realization of the above scheme 
considered in this paper.
We will always take $F$ to be an 
infinite dimensional space 
$F = \bigoplus_{m\ge 0}\Q(q)|m\rangle$ with a generic parameter $q$.
The dual space will be denoted by 
$F^\ast = \bigoplus_{m\ge 0}\Q(q)\langle m|$
with the bilinear pairing 
\begin{align}\label{dup}
\langle m |n\rangle = (q^2)_m\delta_{m,n}.
\end{align}

The solution $\Rm$ of the tetrahedron equation 
we are concerned with is the one 
obtained as the intertwiner of the quantum coordinate ring 
$A_q(sl_3)$ \cite{KV}\footnote{
The formula for it on p194 in \cite{KV} contains a misprint unfortunately.
Eq. (\ref{Rex}) here is a correction of it.},  
which was also found 
from a quantum geometry consideration in a different gauge including 
square roots \cite{BS, BMS}. 
They were shown to be essentially the same object 
and to constitute the solution of the 3d reflection equation in \cite{KO1}.
It can also be identified with the transition matrix
of the PBW bases of the nilpotent subalgebra of $U_q(sl_3)$ \cite{S,KOY}. 
Here we simply call it 3d $R$. 
It is given by
\begin{align}
&\Rm(|i\rangle \otimes |j\rangle \otimes |k\rangle) = 
\sum_{a,b,c\ge 0} \Rm^{a,b,c}_{i,j,k}
|a\rangle \otimes |b\rangle \otimes |c\rangle,\label{Rabc}\\
&\Rm^{a,b,c}_{i,j,k} =\delta^{a+b}_{i+j}\delta^{b+c}_{j+k}
\sum_{\lambda+\mu=b}(-1)^\lambda
q^{i(c-j)+(k+1)\lambda+\mu(\mu-k)}
\frac{(q^2)_{c+\mu}}{(q^2)_c}
\binom{i}{\mu}_{\!\!q^2}
\binom{j}{\lambda}_{\!\!q^2},\label{Rex}
\end{align}
where $\delta^m_{n}=\delta_{m,n}$ just to save the space.
The sum (\ref{Rex}) is over $\lambda, \mu \ge 0$ 
satisfying $\lambda+\mu=b$, which is also bounded by the 
condition $\mu\le i$ and $\lambda \le j$.
The formula (\ref{Rex}) is taken from \cite[eq.(2.20)]{KO1}.
The fact that 
\begin{align}\label{cl0}
\Rm^{a,b,c}_{i,j,k}=0\;\; \text{unless} \;\;(a+b,b+c)=(i+j,j+k)
\end{align}
plays an important role and will be refereed to as {\em conservation law}.
Further properties of $\Rm$ have been summarized in Appendix \ref{app:3dR}.
It is natural to depict (\ref{Rabc}) as follows:
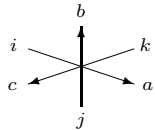
\[
\begin{picture}(80,45)(-40,-22)

\put(0,0){\vector(0,1){15}}\put(0,0){\line(0,-1){15}}
\put(-2,19){$\scriptstyle{b}$} \put(-2,-22){$\scriptstyle{j}$}

\put(0,0){\line(3,1){20}}\put(0,0){\vector(-3,-1){20}}
\put(22,6){$\scriptstyle{k}$} \put(-28,-9){$\scriptstyle{c}$}

\put(0,0){\line(-3,1){20}}\put(0,0){\vector(3,-1){20}}
\put(23,-9){$\scriptstyle{a}$}\put(-27,6){$\scriptstyle{i}$}

\end{picture}
\] 

Let us turn to the vectors $|\chi_s(z)\rangle$ and 
$\langle \chi_s(z)|$ in (\ref{cv})--(\ref{XR}).
We use two such vectors obtained in \cite{KS}.
In the present notation they read
\begin{align}
&|\chi_1(z)\rangle = \sum_{m\ge 0}\frac{z^m}{(q)_m}|m\rangle,\quad
|\chi_2(z)\rangle = \sum_{m\ge 0}\frac{z^m}{(q^4)_m}|2m\rangle,
\label{xk}\\
&\langle \chi_1(z)| = \sum_{m\ge 0}\frac{z^m}{(q)_m}\langle m|,\quad
\langle \chi_2(z)| = \sum_{m\ge 0}\frac{z^m}{(q^4)_m}\langle 2m|.
\label{xb}
\end{align}

\subsection{\mathversion{bold}Solution $S^{s,t}(z)$ to the Yang-Baxter equation}\label{ss:S}
We define the four families of solutions to the Yang-Baxter equation
$S(z) = S^{s,t}(z)= S^{s,t}(z,q)$ $(s,t=1,2)$ by the formula
(\ref{sdef}) by substituting  (\ref{Rex}), (\ref{xk}) and (\ref{xb})
into it.
Each family consists of the solutions 
labeled with $n \in \Z_{\ge 1}$. 
They are the matrices acting on $F^{\otimes n} \otimes F^{\otimes n}$ 
whose elements are given by \cite[Remark 1]{KO3}
\begin{align}
&S^{s,t}(z)\bigl(|{\bf i}\rangle \otimes |{\bf j}\rangle\bigr)
= \sum_{{\bf a},{\bf b}}
S^{s,t}(z)^{{\bf a},{\bf b}}_{{\bf i},{\bf j}}
|{\bf a}\rangle \otimes |{\bf b}\rangle,\label{sact}\\
&S^{s,t}(z)^{{\bf a},{\bf b}}_{{\bf i},{\bf j}}
=\varrho^{s,t}(z)\!\!\!\sum_{c_0, \ldots, c_n\ge 0} 
\frac{z^{c_0}(q^2)_{sc_0}}{(q^{s^2})_{c_0}(q^{t^2})_{c_n}}
\Rm^{a_1, b_1, sc_0}_{i_1, j_1, c_1}
\Rm^{a_2, b_2, c_1}_{i_2, j_2, c_2}\cdots
\Rm^{a_{n\!-\!1}, b_{n\!-\!1}, c_{n\!-\!2}}_{i_{n\!-\!1}, j_{n\!-\!1}, c_{n\!-\!1}}
\Rm^{a_n, b_n, c_{n\!-\!1}}_{i_n, j_n, tc_n},\label{sabij}
\end{align}
where $|{\bf a}\rangle = 
|a_1\rangle \otimes \cdots \otimes |a_n\rangle \in F^{\otimes n}$
for ${\bf a} = (a_1,\ldots, a_n) \in (\Z_{\ge 0})^n$, etc.
The factor $(q^2)_{sc_0}$ originates in (\ref{dup}).
By Applying (\ref{rtb}) to (\ref{sabij}) it is straightforward to show
\begin{align}\label{s21}
S^{t,s}(z)^{{\bf a},{\bf b}}_{{\bf i},{\bf j}}/\varrho^{t,s}(z)
=\left(\prod_{r=1}^n\frac{z^{\frac{1}{t}j_r}(q^2)_{i_r}(q^2)_{j_r}}
{z^{\frac{1}{t}b_r}(q^2)_{a_r}(q^2)_{b_r}}\right)
S^{s,t}(z^{\frac{s}{t}})^{\overline{\bf i},
\overline{\bf j}}_{\overline{\bf a},\overline{\bf b}}
/\varrho^{s,t}(z^{\frac{s}{t}}),
\end{align}
where $\overline{\bf a} = (a_n,\ldots, a_1)$ is the reverse array of  
${\bf a} = (a_1,\ldots, a_n)$ and similarly for 
$\overline{\bf b}, \overline{\bf i}$ and $\overline{\bf j}$.
Henceforth we shall only consider
$S^{1,1}(z), S^{1,2}(z)$ and $S^{2,2}(z)$ in the rest of the paper.
The matrix element (\ref{sabij}) is depicted as follows:
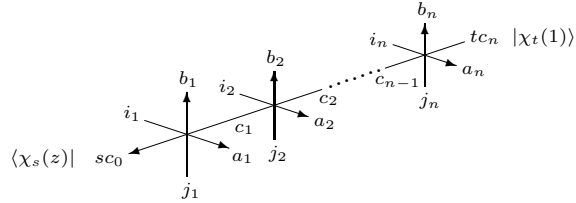
\[
\begin{picture}(200,75)(-100,-40)
\put(3,1){\vector(-3,-1){73}}

\put(-115,-27){$\scriptstyle{\langle \chi_s(z) |}$}

\put(-83,-27){$\scriptstyle{sc_0}$}
 
\put(-48,-16){\vector(0,1){16}}\put(-48,-16){\line(0,-1){16}}
\put(-48,-16){\vector(3,-1){16}} \put(-48,-16){\line(-3,1){16}}
\put(-51,3){$\scriptstyle{b_1}$}
\put(-72,-10){$\scriptstyle{i_1}$}
\put(-31,-26){$\scriptstyle{a_1}$}
\put(-50,-39){$\scriptstyle{j_1}$}

\put(-30,-15){$\scriptstyle{c_1}$}

\put(-15,-5){\vector(0,1){13}}\put(-15,-5){\line(0,-1){13}}
\put(-15,-5){\vector(3,-1){13}}\put(-15,-5){\line(-3,1){13}}
\put(-36,0){$\scriptstyle{i_2}$}
\put(-18,11){$\scriptstyle{b_2}$}
\put(-17,-25){$\scriptstyle{j_2}$}
\put(0,-13){$\scriptstyle{a_2}$}

\put(2,-4){$\scriptstyle{c_2}$}

\multiput(5.1,1.7)(3,1){7}{.} 
\put(6,2){
\put(21,7){\line(3,1){30}}
\put(36,12){\vector(0,1){12}}\put(36,12){\line(0,-1){12}}
\put(36,12){\vector(3,-1){12}}\put(36,12){\line(-3,1){12}}
\put(15,16){$\scriptstyle{i_n}$}
\put(50,5){$\scriptstyle{a_n}$}
\put(33,27){$\scriptstyle{b_n}$}
\put(34,-7){$\scriptstyle{j_n}$}

\put(17,1){$\scriptstyle{c_{n-1}}$}

\put(53,17){$\scriptstyle{tc_n}$}
}
 
\put(75,19){$\scriptstyle{|\chi_t(1) \rangle}$}
 
 \end{picture}
 \]

Due to (\ref{cl0}), 
$S^{s,t}(z)$ also obeys 
the conservation law
\begin{align}\label{claw}
S^{s,t}(z)^{{\bf a},{\bf b}}_{{\bf i},{\bf j}}=0\;\;
\text{unless}\;\;
{\bf a}+{\bf b} = {\bf i} + {\bf j}.
\end{align}
Due to the factor $\delta^{b+c}_{j+k}$ in (\ref{Rex}),
the sum (\ref{sabij}) is constrained by the 
$n$ conditions
$b_1+sc_0 = j_1+c_1, \ldots, b_n+c_{n-1}=j_n+tc_n$.
Therefore it is actually a {\em single} sum.
For $(s,t)=(2,2)$, they further enforce
a parity constraint
\begin{align}\label{d22}
S^{2,2}(z)^{{\bf a},{\bf b}}_{{\bf i},{\bf j}}=0
\;\;
\text{unless}\;\;
|{\bf a}| \equiv |{\bf i}|,\;\; |{\bf b}| \equiv |{\bf j}| \;\mod 2,
\end{align}
where $|{\bf a}|=a_1+\cdots + a_n$, etc.
Thus we have a direct sum decomposition
\begin{align}
S^{2,2}(z)& 
= S^{+,+}(z)\oplus S^{+,-}(z)\oplus S^{-,+}(z)\oplus S^{-,-}(z),
\label{22pm}\\
S^{\epsilon_1,\epsilon_2}(z)
& \in \mathrm{End}
\bigl((F^{\otimes n})_{\epsilon_1}\otimes 
(F^{\otimes n})_{\epsilon_2}\bigr),
\qquad
(F^{\otimes n})_{\pm} =
\bigoplus_{{\bf a} \in (\Z_{\ge 0})^n,\,
(-1)^{|{\bf a}|}=\pm 1}\Q(q)|{\bf a}\rangle.\label{fpm}
\end{align} 
We dare allow the coexistence of somewhat confusing notations 
$S^{s,t}(z)$ and $S^{\epsilon_1,\epsilon_2}(z)$
expecting that they can be properly distinguished from the context.
(A similar warning applies to $\varrho^{s,t} (z)$ in the sequel.)

We choose the normalization factors as
\begin{align}\label{rst}
\varrho^{1,1} (z)= 
\frac{(z; q)_\infty}{(-zq; q)_\infty},\;\;
\varrho^{1,2} (z)= 
\frac{(z^2; q^2)_\infty}{(-z^2q; q^2)_\infty},\;\;
\varrho^{\epsilon_1,\epsilon_2}(z)= 
\Bigl(\frac{(z; q^4)_\infty}{(zq^2; q^4)_\infty}
\Bigr)^{\epsilon_1\epsilon_2},
\end{align}
which agrees with \cite[eq.(2.22)]{KO3} for $n=1$ case.
Then the matrix elements of $S^{1,1}(z), S^{1,2}(z)$ 
and $S^{\epsilon_1,\epsilon_2}(z)$ are 
rational functions of $q$ and $z$.

\subsection{Examples}\label{ss:ex}
Let us demonstrate the calculations of the matrix elements
$S^{s,t}(z)^{{\bf a},{\bf b}}_{{\bf i},{\bf j}}$ (\ref{sabij}) 
on simple examples.
We pick a few simple matrix elements derivable from (\ref{Rex})
and (\ref{rtb}): 
\begin{align*}
\Rm^{a,0,c}_{i,j,k}&=q^{ik}\delta^{a}_{i+j}\delta^c_{j+k},\quad
\Rm^{a,b,c}_{i,0,k}= q^{ac}
\frac{(q^2)_i(q^2)_k}{(q^2)_a(q^2)_b(q^2)_c}
\delta^{a+b}_{i}\delta^{b+c}_{k},
\quad \Rm^{1,1,k}_{1,1,k}=1\!-\!(1\!+\!q^2)q^{2k},\\
\Rm^{a,b,c}_{0,j,k}&=(-1)^bq^{b(k+1)}\binom{j}{b}_{\!\!\!q^2}
\delta^{a+b}_{j}\delta^{b+c}_{j+k},
\quad
\Rm^{0,b,c}_{i,j,k}=
(-1)^jq^{j(c+1)}\frac{(q^2)_k}{(q^2)_c}
\delta^{b}_{i+j}\delta^{b+c}_{j+k}.
\end{align*}
Using them we find
\begin{align}
&S^{1,1}(z) (|{\bf 0}\rangle \otimes |{\bf 0}\rangle)=
S^{1,2}(z) (|{\bf 0}\rangle \otimes |{\bf 0}\rangle)=
S^{+,+}(z) (|{\bf 0}\rangle \otimes |{\bf 0}\rangle)=
|{\bf 0}\rangle \otimes |{\bf 0}\rangle,\label{s0}\\
&S^{+,-}(z)^{{\bf 0},{\bf e}_1}_{{\bf 0},{\bf e}_1}=\frac{-q}{1-z},\quad
S^{-,+}(z)^{{\bf e}_1,{\bf 0}}_{{\bf e}_1,{\bf 0}}=\frac{1}{1-z},\quad
S^{-,-}(z)^{{\bf e}_1,{\bf e}_1}_{{\bf e}_1,{\bf e}_1}
=\frac{z-q^2}{1-zq^2},
\label{spm}
\end{align}
where ${\bf 0} = (0,\ldots, 0)$ and 
${\bf e}_i = (0,\ldots,0,\overset{i}{1},0,\ldots, 0) \in  \Z^n$.
In fact for any ${\bf a}=(a_1,\ldots, a_n) \in (\Z_{\ge 0})^n$,
the formulas  
\begin{align*}
S^{1,t}(z)^{{\bf a},{\bf 0}}_{{\bf a},{\bf 0}}
&=
(-q)^{-|{\bf a}|}S^{1,t}(z)^{{\bf 0},{\bf a}}_{{\bf 0},{\bf a}}
=\frac{(z^t; q^t)_{|{\bf a}|}}{(-z^tq; q^t)_{|{\bf a}|}}
\quad(t=1,2),\\
S^{+,+}(z)^{{\bf a},{\bf 0}}_{{\bf a},{\bf 0}}
&=(-q)^{-|{\bf a}|}S^{+,+}(z)^{{\bf 0},{\bf a}}_{{\bf 0},{\bf a}}
=\frac{(z; q^4)_{|{\bf a}|/2}}{(zq^2; q^4)_{|{\bf a}|/2}}\quad 
(|{\bf a}|\in 2\Z)
\end{align*}
are valid.
We also have
\begin{align*}
&S^{1,1}(z)^{2{\bf e}_1,{\bf 0}}_{{\bf e}_1,{\bf e}_1}
=(-q)^{-1}
S^{1,1}(z)^{{\bf 0},2{\bf e}_1}_{{\bf e}_1,{\bf e}_1}
=\frac{(1+q)(1-z)}{(1+zq)(1+zq^2)},\\
&S^{-,-}(z)^{{\bf e}_n,{\bf e}_n}_{{\bf e}_n,{\bf e}_n}
=S^{-,-}(z)^{{\bf e}_1,{\bf e}_1}_{{\bf e}_1,{\bf e}_1},\quad
S^{-,-}(z)^{{\bf e}_n,{\bf e}_1}_{{\bf e}_1,{\bf e}_n}
=z^{-1}S^{-,-}(z)^{{\bf e}_1,{\bf e}_n}_{{\bf e}_n,{\bf e}_1}
=\frac{1-q^2}{1-zq^2},\\
&S^{-,-}(z)^{{\bf e}_1,{\bf e}_n}_{{\bf e}_1,{\bf e}_n}
=S^{-,-}(z)^{{\bf e}_n,{\bf e}_1}_{{\bf e}_n,{\bf e}_1}
=-\frac{q(1-z)}{1-zq^2}.
\end{align*}
For instance to derive the last result in (\ref{spm}),
one looks at the  corresponding sum (\ref{sabij}):
\begin{align*}
\sum_{c_0, \ldots, c_n\ge 0} 
\frac{z^{c_0}(q^2)_{2c_0}}{(q^{4})_{c_0}(q^{4})_{c_n}}
\Rm^{1, 1, 2c_0}_{1, 1, c_1}
\Rm^{0, 0, c_1}_{0, 0, c_2}\cdots
\Rm^{0,0, c_{n\!-\!2}}_{0,0, c_{n\!-\!1}}
\Rm^{0,0, c_{n\!-\!1}}_{0,0, 2c_n}.
\end{align*}
Due to (\ref{cl0}) this is a single sum 
over $k=c_0=c_n=c_1/2=\cdots = c_{n-1}/2$.
Moreover the product of $\Rm$'s is equal to 
$\Rm^{1,1,2k}_{1,1,2k}=1-(1+q^2)q^{4k}$. 
Thus it is calculated as 
\begin{align*}
&\sum_{k\ge 0}\frac{z^k(q^2)_{2k}}{(q^4)^2_k}
\bigl(1-(1+q^2)q^{4k}\bigr)=
\sum_{k\ge 0}\frac{z^k(q^2;q^4)_{k}}{(q^4;q^4)_k}
\bigl(1-(1+q^2)q^{4k}\bigr)\\
&=\frac{(zq^2;q^4)_\infty}{(z;q^4)_\infty}
-(1+q^2)\frac{(zq^6;q^4)_\infty}{(zq^4;q^4)_\infty}
=\varrho^{-,-}(z)^{-1}\frac{z-q^2}{1-zq^2}
\end{align*}
by means of the identity \cite[eq.(1.3.12)]{GR}
\begin{align*}
\sum_{k\ge 0}\frac{(x;p)_k}{(p;p)_k}z^k = \frac{(zx;p)_\infty}{(z;p)_\infty}.
\end{align*}
General matrix elements for $n=1$ 
has been obtained in \cite[Proposition 2]{KO3}.

\section{Quantum $R$ matrices for $q$-oscillator representations}\label{sec:R}
\subsection{Quantum affine algebras}
The Drinfeld-Jimbo quantum affine algebras (without derivation operator) 
$U_q=U_q(A^{(2)}_{2n})$, 
$U_q(C^{(1)}_{n})$ and $U_q(D^{(2)}_{n+1})$
are the Hopf algebras 
generated by $e_i, f_i, k^{\pm 1}_i\, (0 \le i \le n)$ satisfying the relations
\begin{equation}\label{uqdef}
\begin{split}
&k_i k^{-1}_i = k^{-1}_i k_i = 1,\quad [k_i, k_j]=0,\\
&k_ie_jk^{-1}_i = q_i^{a_{ij}}e_j,\quad 
k_if_jk^{-1}_i = q_i^{-a_{ij}}f_j,\quad
[e_i, f_j]=\delta_{ij}\frac{k_i-k^{-1}_i}{q_i-q^{-1}_i},\\
&\sum_{\nu=0}^{1-a_{ij}}(-1)^\nu
e^{(1-a_{ij}-\nu)}_i e_j e_i^{(\nu)}=0,
\quad
\sum_{\nu=0}^{1-a_{ij}}(-1)^\nu
f^{(1-a_{ij}-\nu)}_i f_j f_i^{(\nu)}=0\;\;(i\neq j),
\end{split}
\end{equation}
where $e^{(\nu)}_i = e^\nu_i/[\nu]_{q_i}!, \,
f^{(\nu)}_i = f^\nu_i/[\nu]_{q_i}!$.
The data $q_i$ will be 
specified for the algebras under consideration in the sequel.
The Cartan matrix $(a_{ij})_{0 \le i,j \le n}$ \cite{Kac} is given by
\begin{align}\label{car}
a_{i,j} = \begin{cases}
2 & i=j,\\
-\max((\log q_j)/(\log q_i),1) & |i-j|=1,\\
0 & \text{otherwise}.
\end{cases}
\end{align}
We use the coproduct $\Delta$ of the form 
\begin{align*}
\Delta k^{\pm 1}_i = k^{\pm 1}_i\otimes k^{\pm 1}_i,\quad
\Delta e_i = 1\otimes e_i + e_i \otimes k_i,\quad
\Delta f_i = f_i\otimes 1 + k^{-1}_i\otimes f_i.
\end{align*}

\subsection{$\boldsymbol{q}$-oscillator representations}\label{ss:qor}
We introduce representations of
$U_q(D^{(2)}_{n+1})$, $U_q(A^{(2)}_{2n})$
and $U_q(C^{(1)}_{n})$
on the tensor product of the Fock space ${\hat F}^{\otimes n}$
or $F^{\otimes n}$.
Here ${\hat F}= \bigoplus_{m\ge 0}\C(q^{\frac{1}{2}})|m\rangle$
is a slight extension of the coefficient field of 
$F = \bigoplus_{m\ge 0}\Q(q)|m\rangle \subset {\hat F}$.
For  $U_q(A^{(2)}_{2n})$ and $U_q(C^{(1)}_{n})$,
they are essentially the affinization of the $q$-oscillator 
representation of the classical part $U_q(C_n)$ \cite{Ha} 
which factors through the algebra homomorphism 
from $U_q$ to the $q$-Weyl algebra.
A similar feature is expected also for $U_q(D^{(2)}_{n+1})$.
As in the previous section 
we write the elements of ${\hat F}^{\otimes n}$ as
\begin{align}\label{mv}
|{\bf m}\rangle =  |m_1\rangle \otimes  \cdots \otimes |m_n \rangle \in 
{\hat F}^{\otimes n} \quad
\text{for}\;\;
{\bf m} = (m_1,\ldots, m_n) \in (\Z_{\ge 0})^n
\end{align}
and describe the changes in ${\bf m}$ by the vectors 
${\bf e}_i = (0,\ldots,0,\overset{i}{1},0,\ldots, 0) \in  \Z^n$.

Consider  $U_q(D^{(2)}_{n+1})$.
The Dynkin diagram of $D^{(2)}_{n+1}$ looks as
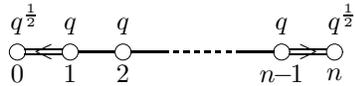
\[
\begin{picture}(126,27)(0,-9)
\multiput( 0,0)(20,0){3}{\circle{6}}
\multiput(100,0)(20,0){2}{\circle{6}}
\multiput(23,0)(20,0){2}{\line(1,0){14}}
\put(83,0){\line(1,0){14}}
\multiput( 2.85,-1)(0,2){2}{\line(1,0){14.3}} 
\multiput(102.85,-1)(0,2){2}{\line(1,0){14.3}} 
\multiput(59,0)(4,0){6}{\line(1,0){2}} 
\put(10,0){\makebox(0,0){$<$}}
\put(110,0){\makebox(0,0){$>$}}
\put(0,-5){\makebox(0,0)[t]{$0$}}
\put(20,-5){\makebox(0,0)[t]{$1$}}
\put(40,-5){\makebox(0,0)[t]{$2$}}
\put(100,-5){\makebox(0,0)[t]{$n\!\! -\!\! 1$}}
\put(120,-6.5){\makebox(0,0)[t]{$n$}}

\put(3,18){\makebox(0,0)[t]{$q^{\frac{1}{2}}$}}
\put(20,13){\makebox(0,0)[t]{$q$}}
\put(40,13){\makebox(0,0)[t]{$q$}}
\put(100,13){\makebox(0,0)[t]{$q$}}
\put(123,18){\makebox(0,0)[t]{$q^{\frac{1}{2}}$}}
\end{picture}
\]
Here the vertices are numbered with $\{0,\ldots, n\}$ as indicated.
The $q_i$ associated to the vertex $i$ is specified above it, so
$q_0=q_n=q^{\frac{1}{2}}$ and $q_j = q$ for $0<j<n$.
The Cartan matrix is given according to (\ref{car}) 
as $a_{0,1}=a_{n,n-1}=-2, a_{1,0}=-1$, etc.
Similar conventions will also be adopted for the other algebras under consideration.
Somewhat unusual convention to include $q^{\frac{1}{2}}$
is to make the presentation and proof of 
our main Theorem \ref{th:main} simple and uniform.  
In Proposition \ref{pr:repD}, \ref{pr:repA} and \ref{pr:repC},
the symbol $[m]$ denotes $[m]_q$.

\begin{proposition}\label{pr:repD}
The following defines an irreducible $U_q(D^{(2)}_{n+1})$ module structure on 
${\hat F}^{\otimes n}$.
\begin{align*}
e_0|{\bf m}\rangle &= x|{\bf m}+{\bf e}_1\rangle,\\
f_0|{\bf m}\rangle  &= i \kappa [m_1] x^{-1}|{\bf m}-{\bf e}_1\rangle,\\
k_0|{\bf m}\rangle  &= -i q^{m_1+\frac{1}{2}}|{\bf m}\rangle,\\
e_j|{\bf m}\rangle &= [m_j]|{\bf m}-{\bf e}_j+{\bf e}_{j+1}\rangle
\quad(1\le  j  \le n-1),\\
f_j|{\bf m}\rangle &= [m_{j+1}]|{\bf m}+{\bf e}_j-{\bf e}_{j+1}\rangle
\quad(1\le  j  \le n-1),\\
k_j|{\bf m}\rangle &= q^{-m_{j}+m_{j+1}}|{\bf m}\rangle
\quad(1\le  j  \le n-1),\\
e_n|{\bf m}\rangle & =  i\kappa[m_n]|{\bf m}-{\bf e}_n\rangle,\\
f_n|{\bf m}\rangle &=|{\bf m}+{\bf e}_n\rangle,\\
k_n|{\bf m}\rangle &= iq^{-m_n-\frac{1}{2}}|{\bf m}\rangle,
\end{align*}
where $x$ is a nonzero parameter and 
\begin{align}\label{rdef}
\kappa = \frac{q+1}{q-1}.
\end{align} 

\end{proposition}

Consider $U_q(A^{(2)}_{2n})$.
The Dynkin diagram of $A^{(2)}_{2n}$ looks as
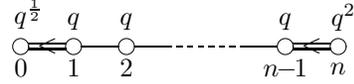
\[
\begin{picture}(126,27)(0,-9)
\multiput( 0,0)(20,0){3}{\circle{6}}
\multiput(100,0)(20,0){2}{\circle{6}}
\multiput(23,0)(20,0){2}{\line(1,0){14}}
\put(83,0){\line(1,0){14}}
\multiput( 2.85,-1)(0,2){2}{\line(1,0){14.3}} 
\multiput(102.85,-1)(0,2){2}{\line(1,0){14.3}} 
\multiput(59,0)(4,0){6}{\line(1,0){2}} 
\put(10,0){\makebox(0,0){$<$}}
\put(110,0){\makebox(0,0){$<$}}
\put(0,-5){\makebox(0,0)[t]{$0$}}
\put(20,-5){\makebox(0,0)[t]{$1$}}
\put(40,-5){\makebox(0,0)[t]{$2$}}
\put(100,-5){\makebox(0,0)[t]{$n\!\! -\!\! 1$}}
\put(120,-6.5){\makebox(0,0)[t]{$n$}}
\put(3,18){\makebox(0,0)[t]{$q^{\frac{1}{2}}$}}
\put(20,12.5){\makebox(0,0)[t]{$q$}}
\put(40,12.5){\makebox(0,0)[t]{$q$}}
\put(100,12.5){\makebox(0,0)[t]{$q$}}
\put(122,16){\makebox(0,0)[t]{$q^2$}}
\end{picture}
\]

\begin{proposition}\label{pr:repA}
The following defines an irreducible $U_q(A^{(2)}_{2n})$ module structure on 
${\hat F}^{\otimes n}$.
\begin{align*}
e_0|{\bf m}\rangle &= x|{\bf m}+{\bf e}_1\rangle,\\
f_0|{\bf m}\rangle  &= i\kappa [m_1]x^{-1}|{\bf m}-{\bf e}_1\rangle,\\
k_0|{\bf m}\rangle  &= -iq^{m_1+\frac{1}{2}}|{\bf m}\rangle,\\
e_j|{\bf m}\rangle &= [m_j]|{\bf m}-{\bf e}_j+{\bf e}_{j+1}\rangle
\quad(1\le  j  \le n-1),\\
f_j|{\bf m}\rangle &= [m_{j+1}]|{\bf m}+{\bf e}_j-{\bf e}_{j+1}\rangle
\quad(1\le  j  \le n-1),\\
k_j|{\bf m}\rangle &= q^{-m_{j}+m_{j+1}}|{\bf m}\rangle
\quad(1\le  j  \le n-1),\\
e_n|{\bf m}\rangle & = \frac{[m_n][m_n-1]}{[2]^2}
|{\bf m}-2{\bf e}_n\rangle,\\
f_n|{\bf m}\rangle &=|{\bf m}+2{\bf e}_n\rangle,\\
k_n|{\bf m}\rangle &=-q^{-2m_n-1}|{\bf m}\rangle,
\end{align*}
where $x$ is a nonzero parameter and 
$\kappa$ is defined by (\ref{rdef}).
\end{proposition}

Consider $U_q(C^{(1)}_{n})$.
The Dynkin diagram of $C^{(1)}_{n}$ looks as
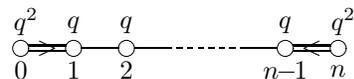
\[
\begin{picture}(126,27)(0,-9)
\multiput( 0,0)(20,0){3}{\circle{6}}
\multiput(100,0)(20,0){2}{\circle{6}}
\multiput(23,0)(20,0){2}{\line(1,0){14}}
\put(83,0){\line(1,0){14}}
\multiput( 2.85,-1)(0,2){2}{\line(1,0){14.3}} 
\multiput(102.85,-1)(0,2){2}{\line(1,0){14.3}} 
\multiput(59,0)(4,0){6}{\line(1,0){2}} 
\put(10,0){\makebox(0,0){$>$}}
\put(110,0){\makebox(0,0){$<$}}
\put(0,-5){\makebox(0,0)[t]{$0$}}
\put(20,-5){\makebox(0,0)[t]{$1$}}
\put(40,-5){\makebox(0,0)[t]{$2$}}
\put(100,-5){\makebox(0,0)[t]{$n\!\! -\!\! 1$}}
\put(120,-6.5){\makebox(0,0)[t]{$n$}}

\put(2,15){\makebox(0,0)[t]{$q^2$}}
\put(20,12){\makebox(0,0)[t]{$q$}}
\put(40,12){\makebox(0,0)[t]{$q$}}
\put(100,12){\makebox(0,0)[t]{$q$}}
\put(122,15){\makebox(0,0)[t]{$q^2$}}
\end{picture}
\]

\begin{proposition}\label{pr:repC}
The following defines an irreducible $U_q(C^{(1)}_{n})$ module structure 
on $(F^{\otimes n})_+$ and $(F^{\otimes n})_-$ defined in (\ref{fpm}).
\begin{align*}
e_0|{\bf m}\rangle &= x|{\bf m}+2{\bf e}_1\rangle,\\
f_0|{\bf m}\rangle  &= \frac{[m_1][m_1-1]}{[2]^2}
 x^{-1}|{\bf m}-2{\bf e}_1\rangle,\\
k_0|{\bf m}\rangle  &= -q^{2m_1+1}|{\bf m}\rangle,\\
e_j|{\bf m}\rangle &= [m_j]|{\bf m}-{\bf e}_j+{\bf e}_{j+1}\rangle
\quad(1\le  j  \le n-1),\\
f_j|{\bf m}\rangle &= [m_{j+1}]|{\bf m}+{\bf e}_j-{\bf e}_{j+1}\rangle
\quad(1\le  j  \le n-1),\\
k_j|{\bf m}\rangle &= q^{-m_{j}+m_{j+1}}|{\bf m}\rangle
\quad(1\le  j  \le n-1),\\
e_n|{\bf m}\rangle & = \frac{[m_n][m_n-1]}{[2]^2}
|{\bf m}-2{\bf e}_n\rangle,\\
f_n|{\bf m}\rangle &=|{\bf m}+2{\bf e}_n\rangle,\\
k_n|{\bf m}\rangle &= -q^{-2m_n-1}|{\bf m}\rangle,
\end{align*}
where $x$ is a nonzero parameter.
\end{proposition}
Direct calculation verifies these are representations of $U_q$. To see the irreducibility
decompose ${\hat F}^{\otimes n}$ as a $U_q(A_{n-1})$ module forgetting the action of
generators indexed by $0$ and $n$. By ${\hat F}^{\otimes n}\simeq\bigoplus_{l=0}^\infty
U_q(A_{n-1})|l{\bf e}_1\rangle$ and considering the action of $e_0$ and $f_0$ the irreducibility follows.
We call the irreducible representations given there 
the {\em $q$-oscillator representations} of $U_q$.
We remark that for the twisted case 
$U_q(D^{(2)}_{n+1})$ and $U_q(A^{(2)}_{2n})$,
they are {\em singular} at $q=1$ because of the factor 
$\kappa$ (\ref{rdef}).

\subsection{Quantum $R$ matrices}\label{ss:qR}
Let $V= \hat{F}^{\otimes n}$ for 
$U_q(D^{(2)}_{n+1}),  U_q(A^{(2)}_{2n})$ and 
$V= F^{\otimes n}$ for $U_q(C^{(1)}_n)$.
First we consider $U_q(D^{(2)}_{n+1})$ and $U_q(A^{(2)}_{2n})$.
Let $V_x = {\hat F}^{\otimes n}[x,x^{-1}]$ 
be the representation space of $U_q$ 
in Propositions \ref{pr:repD} and \ref{pr:repA}.  
By the existence of the universal $R$ matrix \cite{D}
there exists an element  
$R \in \mathrm{End}(V_x\otimes V_y)$ 
such that 
\begin{align}\label{drrd}
\Delta'(g) R = R \Delta(g)\quad 
\forall g \in U_q
\end{align}
up to an overall scalar.
Here $\Delta'$ is the opposite coproduct 
defined by $\Delta' = P \circ \Delta$, where
$P(u \otimes v) = v \otimes u$ is the exchange of the components.
Another useful form of (\ref{drrd}) is 
\begin{align}\label{eqrc}
(\pi_y \otimes \pi_x)\Delta(g) PR = PR\,(\pi_x \otimes \pi_y)\Delta(g)
\quad 
\forall g \in U_q,
\end{align}
where $\pi_x: U_q \rightarrow \mathrm{End } \,V_x$ 
denotes the representation.

A little inspection of our representations shows that 
$R$ depends on $x$ and $y$ only through the ratio $z=x/y$.
Moreover, by the irreducibility of $V_x\ot V_y$ (Proposition 
\ref{prop:irred}) $R$ is determined only by postulating (\ref{drrd}) 
for $g=k_r,e_r$ and $f_r$ with $0 \le r \le n$.
Thus denoting the $R$ by $R(z)$, 
we may claim \cite{Ji} that it is determined by the conditions 
\begin{align}
(k_r \otimes k_r)R(z) &= R(z)(k_r\otimes k_r),\label{kR}\\
(e_r\otimes1 + k_r \otimes e_r) R(z) &= 
R(z)(1\otimes e_r + e_r\otimes k_r),\label{eR}\\
(1\otimes f_r + f_r \otimes k^{-1}_r) R(z) &= 
R(z)(f_r\otimes 1 + k^{-1}_r\otimes f_r)\label{fR}
\end{align}
for $0 \le r \le n$ up to an overall scalar.
We fix the normalization of $R(z)$ by 
\begin{align}\label{rnor}
R(z)(|{\bf 0}\rangle \otimes |{\bf 0}\rangle)
= |{\bf 0}\rangle \otimes |{\bf 0}\rangle,
\end{align}
where $|{\bf 0}\rangle \in {\hat F}^{\otimes n}$ is defined after (\ref{spm}).
We call the intertwiner $R(z)$ 
the {\em quantum $R$ matrix} for $q$-oscillator representation.
It satisfies the Yang-Baxter equation
\begin{align}\label{yber}
R_{12}(x)R_{13}(xy)R_{23}(y)
= R_{23}(y)R_{13}(xy)R_{12}(x).
\end{align}

Next we consider $U_q(C^{(1)}_n)$.
Denote by $V^{\pm}_x= (F^{\otimes n})_{\pm}[x,x^{-1}]$
the representation spaces in Proposition \ref{pr:repC}
and set 
$V_x = V^+_x \oplus V^-_x = F^{\otimes n}[x,x^{-1}]$.
See (\ref{fpm}) for the definition of $(F^{\otimes n})_{\pm}$.
We define the quantum $R$ matrix
$R(z)$ to be the direct sum 
\begin{align}\label{Rdeco}
R(z) = R^{+,+}(z)\oplus 
R^{+,-}(z)\oplus 
R^{-,+}(z)\oplus 
R^{-,-}(z),
\end{align}
where each 
$R^{\epsilon_1,\epsilon_2}(z)\in 
\mathrm{End}(V^{\epsilon_1}_x\otimes V^{\epsilon_2}_y)$ is
the quantum $R$ matrix with the 
normalization condition 
\begin{equation}\label{Rpmnor}
\begin{split}
&R^{+,+}(z) (|{\bf 0}\rangle \otimes |{\bf 0}\rangle)
=|{\bf 0}\rangle \otimes |{\bf 0}\rangle,\quad
\qquad\quad
R^{+,-}(z)(|{\bf 0}\rangle \otimes |{\bf e}_1\rangle)
=\frac{-iq^{1/2}}{1-z}|{\bf 0}\rangle \otimes |{\bf e}_1\rangle,\\
&R^{-,+}(z)(|{\bf e}_1\rangle \otimes |{\bf 0}\rangle)=
\frac{-iq^{1/2}}{1-z}|{\bf e}_1\rangle \otimes |{\bf 0}\rangle,\quad
R^{-,-}(z)(|{\bf e}_1\rangle \otimes |{\bf e}_1\rangle)
=\frac{z-q^2}{1-zq^2}
|{\bf e}_1\rangle \otimes |{\bf e}_1\rangle.
\end{split}
\end{equation}
The $R$ matrix $R(z)$
satisfies the Yang-Baxter equation (\ref{yber}).
In fact it is decomposed into the finer equalities 
($\epsilon_1, \epsilon_2, \epsilon_3=\pm$)
\begin{align*}
R^{\epsilon_1,\epsilon_2}_{12}(x)
R^{\epsilon_1,\epsilon_3}_{13}(xy)
R^{\epsilon_2,\epsilon_3}_{23}(y)
= 
R^{\epsilon_2,\epsilon_3}_{23}(y)
R^{\epsilon_1,\epsilon_3}_{13}(xy)
R^{\epsilon_1,\epsilon_2}_{12}(x).
\end{align*}

\subsection{Singular vectors}\label{ss:sv}

In this subsection we find all singular vectors in $V^{\ot 2}$, namely those 
$v \in V^{\otimes 2}$ killed by $e_1,\ldots, e_n$, as a $U_q(B_n)$ module or
$U_q(C_n)$ module. Since $V^{\ot 2}$ is not finite-dimensional, we cannot
say at this stage that they are actually the highest weight vectors of irreducible
modules, but we will see it later in the next subsection. 

\begin{proposition} \label{prop:Bsing}
As a $U_q(B_n)$ module any singular vector in $V^{\ot 2}$ is given by
\[
v_l=\sum_{k=0}^l i^kq^{lk-k^2/2}{l\brack k}
\ket{k{\bf e}_n}\ot\ket{(l-k){\bf e}_n}
\]
for some $l\in\Z_{\ge0}$.
\end{proposition}

\begin{proof}
One can assume that a singular vector $v$ is a weight vector of weight
${\bf l}=\sum_{j=1}^nl_j{\bf e}_j\in(\Z_{\ge0})^n$. Hence $v$ can be written as
\begin{equation} \label{eq0}
v=\sum_{\bf m}c_{\bf m}\ket{\bf m}\ot\ket{{\bf l}-{\bf m}}.
\end{equation}
The conditions $e_jv=0$ ($1\le j\le n-1$) and $e_nv=0$ read respectively as
\begin{align}
\sum c_{\bf m}(&[l_j-m_j]\ket{{\bf m}}\ot\ket{{\bf l}-{\bf m}-{\bf e}_j+{\bf e}_{j+1}}
\nonumber\\
&+q^{-l_j+l_{j+1}+m_j-m_{j+1}}[m_j]
\ket{{\bf m}-{\bf e}_j+{\bf e}_{j+1}}\ot\ket{{\bf l}-{\bf m}})=0, \label{eq1}\\
\sum c_{\bf m}(&[l_n-m_n]\ket{{\bf m}}\ot\ket{{\bf l}-{\bf m}-{\bf e}_n}
\nonumber\\
&+iq^{-l_n+m_n-1/2}[m_n]
\ket{{\bf m}-{\bf e}_n}\ot\ket{{\bf l}-{\bf m}})=0. \label{eq2}
\end{align}

We first show there is no singular vector of weight ${\bf l}$ such that $l_j>0,
l_{j+1}=0$ for some $j<n$. Suppose $l_{j+1}=0$. Looking at the coefficient of
$\ket{\ldots,m_j,0,\ldots}\ot\ket{\ldots,l_j-m_j-1,1,\ldots}$ in \eqref{eq1}
one sees $c_{\bf m}=0$ if $m_j<l_j$. Similarly,  the coefficient of
$\ket{\ldots,m_j-1,1,\ldots}\ot\ket{\ldots,l_j-m_j,0,\ldots}$ in \eqref{eq1}
gives $c_{\bf m}=0$ if $m_j>0$. Hence $c_{\bf m}=0$ for all ${\bf m}$
unless $l_j=0$.

We next show there is no singular vector of weight ${\bf l}$ such that $l_{n-1}>0$.
Looking at the coefficient of $\ket{\ldots,m_{n-1},0}\ot\ket{\ldots,l_{n-1}-m_{n-1}-1,
l_n+1}$ in $\eqref{eq1} _{j=n-1}$ 
one sees $c_{\bf m}=0$ if $m_{n-1}<l_{n-1},m_n=0$. Together with
\eqref{eq2} we get $c_{\bf m}=0$ if $m_{n-1}<l_{n-1}$. 
Similarly,  the coefficient of
$\ket{\ldots,m_{n-1}-1,l_n+1}\ot\ket{\ldots,l_{n-1}-m_{n-1},0}$ 
in $\eqref{eq1} _{j=n-1}$
gives $c_{\bf m}=0$ if $m_{n-1}>0,m_n=l_n$. Together with \eqref{eq2} we get
$c_{\bf m}=0$ for $m_{n-1}>0$. Hence $c_{\bf m}=0$ for all ${\bf m}$
unless $l_{n-1}=0$.
The coefficient 
$c_{\bf m}$ for ${\bf l}=l{\bf e}_n$ can easily be determined by solving \eqref{eq2}.
\end{proof}

\begin{proposition}
As a $U_q(C_n)$ module any singular vector in $V^{\ot 2}$ is given by
\begin{align*}
v^\epsilon_l&=\sum_{0\le k\le l\atop k\equiv p(\epsilon)\,\mbox{\scriptsize mod}\,2} 
q^{k(2l-k-1)/2}{l\brack k}\ket{k{\bf e}_n}\ot\ket{(l-k){\bf e}_n}
\in V^\epsilon\ot V^{\epsilon\eta(l)} \quad(l\ge p(\epsilon)),\\
\intertext{or}
v^-_0&=\ket{{\bf e}_{n-1}}\ot\ket{{\bf e}_n}-q\ket{{\bf e}_n}\ot\ket{{\bf e}_{n-1}}\in V^-\ot V^-
\end{align*}
for some $l\in\Z_{\ge0}$ and $\epsilon=+$ or $-$, where $p(\epsilon)=0$ $(\epsilon=+)$,
\,$=1$ $(\epsilon=-)$ and $\eta(l)=+$ $(l$:even$)$,\,$=-$ $(l$:odd$)$.
\end{proposition}

\begin{proof}
The proof is similar to Proposition \ref{prop:Bsing}. A singular vector $v$ can
be written as \eqref{eq0}. The condition $e_jv=0$ reads \eqref{eq1} and
\begin{align}
\sum c_{\bf m}(&[l_n-m_n][l_n-m_n-1]\ket{{\bf m}}\ot\ket{{\bf l}-{\bf m}-2{\bf e}_n}
\nonumber\\
&-q^{-2l_n+2m_n-1}[m_n][m_n-1]
\ket{{\bf m}-2{\bf e}_n}\ot\ket{{\bf l}-{\bf m}})=0. \label{eq3}
\end{align}

The fact that there is no singular vector 
of weight ${\bf l}$ such that $l_j>0,l_{j+1}=0$ for some $j<n$ is the same.
Similarly to the next argument in the proof of Proposition \ref{prop:Bsing} 
with \eqref{eq3}, we see that $c_{\bf m}=0$ if $m_{n-1}\le l_{n-1}-1$
\& $m_n\equiv0$, $m_{n-1}\le l_{n-1}-2$ \& $m_n\equiv1$, $m_{n-1}\ge1$
\& $m_n\equiv l_n$, or $m_{n-1}\ge2$ \& $m_n\equiv l_n-1$, where the congruence
$\equiv$ is modulo 2. Hence we can conclude that if nontrivial $c_{\bf m}$ exists,
then $l_{n-1}\le1$ if $l_n$ is odd and $l_{n-1}=0$ if $l_n$ is even.

We wish to show that nontrivial $c_{\bf m}$ exists only when ${\bf l}=l{\bf e}_n$
for some nonnegative integer $l$ or ${\bf l}={\bf e}_{n-1}+{\bf e}_n$. 
Thus the remaining thing to show is that if $l_{n-1}=1$ and $l_n$ is odd,
then (i) $l_{n-2}=0$ and (ii) $l_n=1$. To show (i) by contradiction suppose $l_{n-2}>0$.
From the result of the previous paragraph we know $c_{\bf m}=0$ if $m_{n-1}=0$ \& $m_n\equiv0$
or $m_{n-1}=1$ \& $m_n\equiv1$. Hence it suffices to show $c_{\bf m}=0$ when 
(a) $m_{n-2}>0$ \& $m_{n-1}=0$ \& $m_n\equiv1$ (b) $m_{n-2}=m_{n-1}=0$ \& $m_n\equiv1$
(c) $m_{n-2}<l_{n-2}$ \& $m_{n-1}=1$ \& $m_n\equiv0$ (d) $m_{n-2}=l_{n-2}$ \& $m_{n-1}=1$ \& $m_n\equiv0$.
Case (a) (resp. (c)) is shown from \eqref{eq1}$_{j=n-2}$ and $c_{(\ldots,m_{n-2}-1,1,m_n)}=0$ if $m_n\equiv1$
(resp.  $c_{(\ldots,m_{n-2}+1,0,m_n)}=0$ if $m_n\equiv0$). Case (b) (resp. (d)) is shown by comparing
the coefficient of $\ket{\ldots,0,0,m_n}\ot\ket{\ldots,l_{n-2}-1,2,l_n-m_n}$ (resp. 
$\ket{\ldots,l_{n-2}-1,2,m_n}\ot\ket{\ldots,0,0,l_n-m_n}$) in \eqref{eq1}$_{j=n-2}$.
To show (ii) by contradiction suppose $l_n\ge3$. One can compute 
$c_{{\bf e}_{n-1}+2{\bf e}_n}/c_{{\bf e}_n}$ either by 
$(c_{3{\bf e}_n}/c_{{\bf e}_n})(c_{{\bf e}_{n-1}+2{\bf e}_n}/c_{3{\bf e}_n})$ or 
$(c_{{\bf e}_{n-1}}/c_{{\bf e}_n})(c_{{\bf e}_{n-1}+2{\bf e}_n}/c_{{\bf e}_{n-1}})$.
The former gives $-q^{l_n-3}[l_n-1][l_n-2]/([3][2])$ while the latter
$-q^{l_n-3}[l_n][l_n-1]/([2][1])$, which is a contradiction.

For the survived cases it is easy to obtain nontrivial $c_{\bf m}$.
\end{proof}

\subsection{Spectral decomposition}\label{ss:sd}

We calculate the spectral decomposition of $PR(z)$. 
In this subsection we denote the subspace generated by
$f_1,\ldots,f_n$ from the singular vector $v_l$ (resp. $v^\epsilon_l$) by $V_l$ (resp. $V^\epsilon_l$). 
The complete reducibility of $V^{\ot 2}$ as 
a $U_q(B_n)$ or $U_q(C_n)$ module is valid since $PR(z)$ has a different 
eigenvalue on each subspace $V_l$ or $V^\epsilon_l$ as we will see below, 
and the singular vectors obtained in the previous subsection are 
actually highest weight vectors of each irreducible component.

We prepare a lemma, which is obtained by direct calculation.

\begin{lemma}
For $U_q(D_{n+1}^{(2)})$ we have 
\begin{align}
&(\pi_x\ot\pi_y)\Delta(e_{n-1}\cdots e_1e_0)v_l
=\frac1{1-q^{2l+1}}\{(q^{l+1}x+y)v_{l+1}+q^l(x+q^ly)f_n^2v_{l-1}\}\quad(l\ge1),
\label{ev}\\
&(\pi_x\ot\pi_y)\Delta(e_{n-1}\cdots e_1e_0)v_0
=\frac1{1-q}\{(qx+y)v_1-iq^{1/2}(x+y)f_nv_0\},
\label{ev0}\\
&(\pi_x\ot\pi_y)\Delta(f_0f_1\cdots f_{n-1})v_l
=-\kappa[l]q^{-1/2}(q^lx^{-1}+y^{-1})v_{l-1}\quad(l\ge1).
\label{fv}
\end{align}

\end{lemma}

\begin{proposition} \label{prop:specD2}
For $U_q(D_{n+1}^{(2)})$, $PR(z)$ has the following spectral decomposition.
\[
PR(z)=\sum_{l=0}^\infty\prod_{j=1}^l\frac{z+q^j}{1+q^jz}P_l,
\]
where $P_l$ is the projector on $V_l$.
\end{proposition}

\begin{proof}
First we note that by $PR(z)$ $v_l$ is mapped to a scalar multiple of $v_l$
since $PR(z)$ commutes with $U_q(B_n)$. Suppose $PR(z)v_l=\rho_l(z)v_l$.
Setting $g=f_0f_1\cdots f_{n-1}$ in (\ref{eqrc}), applying both 
sides to $v_l$ and using \eqref{fv} we get
\[
(q^ly^{-1}+x^{-1})\rho_l(z)v_{l-1}=\rho_{l-1}(z)(q^lx^{-1}+y^{-1})v_{l-1}.
\]
Due to $v_0=|{\bf 0}\rangle \otimes |{\bf 0}\rangle$
the normalization condition (\ref{rnor}) is satisfied by choosing
$\rho_0(z)=1$. Thus we obtain
\[
\rho_l(z)=\prod_{j=1}^l\frac{z+q^j}{1+q^jz}.
\]
\end{proof}

\begin{lemma}
For $U_q(A_{2n}^{(2)})$ we have
\begin{align*}
&(\pi_x\ot\pi_y)\Delta(e_{n-1}\cdots e_1e_0)(v_l^+\pm v_l^-)
=\frac{[l]}{[2l]}\{q^{-l}(y\mp iq^{l+1/2}x)(v_{l+1}^+\pm v_{l+1}^-)\\
&\hspace{6cm}\mp iq^{-1/2}(x\mp iq^{l-1/2}y)f_n(v_{l-1}^+\pm\chi(l\ne1)v_{l-1}^-)\}
\quad(l\ge1),\\
&(\pi_x\ot\pi_y)\Delta(e_{n-1}\cdots e_1e_0)v^+_0
=yv^+_1-iq^{1/2}xv^-_1,\\
&(\pi_x\ot\pi_y)\Delta(e_{n-2}\cdots e_0e_{n-1}\cdots e_0)v^+_0
=\frac1{[2]^2}\{(q^{-1}y^2-x^2)f_{n-1}v_2^+-i[2]xy f_{n-1}v_2^-\\
&\hspace{7cm}-(x^2+y^2)f_{n-1}f_nv_0^++i[2](q^{1/2}-q^{-1/2})xyv_0^-\},\\
&(\pi_x\ot\pi_y)\Delta(e_{n-1}e_ne_{n-1}\cdots e_0)v^-_0
=\frac1{[2]}(-iq^{1/2}xv^+_1+yv^-_1),\\
&(\pi_x\ot\pi_y)\Delta(f_0f_1\cdots f_{n-1})(v^+_l\pm v^-_l)\\
&\hspace{4.7cm}
=-\kappa[l]q^{-1/2}(\mp iq^{l-1/2}x^{-1}+y^{-1})(v^+_{l-1}\pm \chi(l\ne1)v^-_{l-1})
\quad(l\ge1).
\end{align*}
Here $\chi(\theta)=1$ if $\theta$ is true, $=0$ otherwise.
\end{lemma}

\begin{proposition} \label{prop:specA2}
For $U_q(A_{2n}^{(2)})$, $PR(z)$ has the following spectral decomposition.
\[
PR(z)=P_0+\sum_{l=1}^\infty\left(\prod_{j=1}^l\frac{z-iq^{j-1/2}}{1-iq^{j-1/2}z}P^+_l
+\prod_{j=1}^l\frac{z+iq^{j-1/2}}{1+iq^{j-1/2}z}P^-_l\right),
\]
where $P_0$ is the projector on $V^+_0\oplus V^-_0$ and $P^\pm_l$ with $l\ge1$ is the one on the 
$U_q(C_n)$ invariant subspace containing $v^+_l\pm v^-_l$.
\end{proposition}

\begin{proof}
Set
\[
PR(z)=\sum_{l=0}^\infty\left(
\rho^+_l(z)P^+_l
+\rho^-_l(z)P^-_l\right)
\]
with $\rho^+_0(z)=1$. $P^\epsilon_0$ is the projector on $V^\epsilon_0$.
The determination of $\rho^\epsilon_l(z)$ is similar to Proposition \ref{prop:specD2}. 
The normalization condition (\ref{rnor}) is satisfied since 
$v^+_0=|{\bf 0}\rangle \otimes |{\bf 0}\rangle$.
\end{proof}

\begin{lemma}
For $U_q(C_n^{(1)})$ we have
\begin{align*}
&(\pi_x\ot\pi_y)\Delta((e_{n-1}\cdots e_1)^2e_0)v^\epsilon_l
=\frac{[2]}{\langle l+1\rangle\langle l\rangle\langle l-1\rangle}
\{q^{-2l-1}\langle l-1\rangle(y-q^{2l+2}x)v^\epsilon_{l+2}\\
&\hspace{3cm}-[2]\langle l\rangle(x+y)f_nv^\epsilon_l-q^{-1}\langle l+1\rangle
(x-q^{2l-2}y)f_n^2v^\epsilon_{l-2}\}\quad(l\ge2,(\epsilon,l)\ne(-,2)),\\
&(\pi_x\ot\pi_y)\Delta((e_{n-1}\cdots e_1)^2e_0)v^\epsilon_1
=\frac1{\langle2\rangle}\{q^{-3}(y-q^4x)v^\epsilon_3
-([2-\epsilon1]x+[2+\epsilon1]y)f_nv^\epsilon_1\},\\
&(\pi_x\ot\pi_y)\Delta(e_{n-1}^2e_ne_{n-1}e_{n-2}^2\cdots e_1^2e_0)v^-_0
=q^{-1}[2]^{n-2}(y-q^2x)v^-_2,\\
&(\pi_x\ot\pi_y)\Delta(f_0(f_1\cdots f_{n-1})^2)v^\epsilon_l
=q^{-1}\frac{[l][l-1]}{[2]}(q^{2l-2}x^{-1}-y^{-1})v^\epsilon_{l-2}
\quad(l\ge2,(\epsilon,l)\ne(-,2)),\\
&(\pi_x\ot\pi_y)\Delta(f_0f_1^2\cdots f_{n-2}^2f_{n-1}f_nf_{n-1}^2)v^-_2
=[2]^{n-2}\{(y^{-1}+x^{-1})f_{n-1}v^-_2+q[2](y^{-1}-q^2x^{-1})v^-_0\}.
\end{align*}
Here $\langle m\rangle=q^m+q^{-m}$.
\end{lemma}

\begin{proposition}
For $U_q(C_n^{(1)})$, $PR(z)$ has the following spectral decomposition.
\[
PR^{\epsilon,\epsilon}(z)=\sum_{l=0}^\infty\prod_{j=1}^l\frac{z-q^{4j-2}}{1-q^{4j-2}z}P^{\epsilon}_{2l},
\qquad\quad
PR^{\epsilon,-\epsilon}(z)=
\frac{-iq^{1/2}}{1-z}
\sum_{l=0}^\infty\prod_{j=1}^l\frac{z-q^{4j}}{1-q^{4j}z}P^{\epsilon}_{2l+1},
\]
where $P^{\epsilon}_{2l}$ is the projector on 
$V^\epsilon_{2l}$, and $P^\epsilon_{2l+1}$ is the 
$U_q(C_n)$ linear map sending 
$v^\epsilon_{2l+1}$ to $v^{-\epsilon}_{2l+1}$ and other singular 
vectors to $0$.
\end{proposition}

\begin{proof}
Set
\[
PR^{\epsilon,\epsilon}(z)=\sum_{l=0}^\infty\rho^\epsilon_{2l}(z)P^{\epsilon}_{2l},
\qquad\quad
PR^{\epsilon,-\epsilon}(z)=\sum_{l=0}^\infty\rho^\epsilon_{2l+1}(z)P^{\epsilon}_{2l+1}
\]
with $\rho^\epsilon_0(z)=1$ and 
$\rho^\epsilon_1(z)=\frac{-iq^{1/2}}{1-z}$. 
The necessary data to derive the recursion relations of 
$\rho^\epsilon_l(z)$ are given in the lemma.
The four vectors 
$|{\bf 0}\rangle \otimes |{\bf 0}\rangle,
|{\bf 0}\rangle \otimes |{\bf e}_1\rangle,
|{\bf e}_1\rangle \otimes |{\bf 0}\rangle$ and 
$|{\bf e}_1\rangle \otimes |{\bf e}_1\rangle$
in (\ref{Rpmnor}) are contained in the  
irreducible components generated from 
$v_0^+, v_1^+, v_1^-$ and $v_2^-$, respectively.
Thus the condition (\ref{Rpmnor}) agrees with the 
above normalization of the eigenvalues.
\end{proof}

Finally we prove

\begin{proposition} \label{prop:irred}
As a $U_q(D_{n+1}^{(2)})$ or $U_q(A_{2n}^{(2)})$ module $V_x\ot V_y$ 
is irreducible. As a $U_q(C_n^{(1)})$ module each $V_x^{\epsilon_1}\ot
V_y^{\epsilon_2}$ $(\epsilon_1,\epsilon_2=\pm)$ is irreducible.
\end{proposition}

\begin{proof}
We prove the $U_q(D_{n+1}^{(2)})$ case only. Suppose a submodule contains a
nonzero weight vector. One can assume it is a singular vector. Hence it is $v_l$
for some $l\in\Z_{\ge0}$. By \eqref{fv} the submodule contains $v_0$. Then
by \eqref{ev0} it contains a linear combination of $v_1$ and $f_nv_0$. However,
since eigenvalues of $PR(z)$ for $v_1$ and $f_nv_0$ are different by Proposition
\ref{prop:specD2}, the submodule contains $v_1$. Arguing similarly using 
\eqref{ev}, it contains $v_l$ for any $l\ge\Z_{\ge0}$, and the submodule is nothing
but $V_x\ot V_y$.
\end{proof}

\section{$S^{s,t}(z)$ as quantum $R$ matrix}\label{sec:main}
\subsection{Main theorem}
Define the operator $K$ acting on ${\hat F}^{\otimes n}$ by 
\begin{align*}
K|{\bf m}\rangle = (-iq^{\frac{1}{2}})^{m_1+\cdots + m_n}|{\bf m}\rangle.
\end{align*}
See (\ref{mv}) for the notation.
Introduce the gauge transformed quantum $R$ matrix by
\begin{align}\label{gtr}
{\tilde R}(z) = (K^{-1}\otimes 1) R(z)(1\otimes K).
\end{align}
It is easy to see that ${\tilde R}(z)$ 
also satisfies the Yang-Baxter equation (\ref{yber}).

In Section \ref{ss:S} we have constructed the solutions 
$S^{s,t}(z)$ of 
the Yang-Baxter equation from the 3d $R$  in 
(\ref{sact}), (\ref{sabij}) and (\ref{rst}).
In Section \ref{sec:R}
the quantum $R$ matrices for $q$-oscillator representations
of $U_q(D^{(2)}_{n+1})$, 
$U_q(A^{(2)}_{2n})$ and 
$U_q(C^{(1)}_{n})$ 
have been obtained. 
The next theorem, which is the main result of the paper, 
states the precise relation between them.
\begin{theorem}\label{th:main}
Denote by ${\tilde R}_{\mathfrak g}(z)$ the gauge transformed
quantum $R$ matrix (\ref{gtr}) for 
$U_q({\mathfrak g})$.
Then the following equalities hold:
\begin{align*}
S^{1,1}(z) &= {\tilde R}_{D^{(2)}_{n+1}}(z),\\
S^{1,2}(z) &= {\tilde R}_{A^{(2)}_{2n}}(z),\\
S^{2,2}(z) &= {\tilde R}_{C^{(1)}_{n}}(z),
\end{align*}
where the last one means
$S^{\epsilon_1,\epsilon_2}(z) = \tilde{R}^{\epsilon_1,\epsilon_2}(z)$
between (\ref{22pm}) and (\ref{Rdeco}) with the 
gauge transformation (\ref{gtr}).
\end{theorem}
For $S^{2,1}(z)$, see (\ref{s21}).

\begin{remark}\label{re:dyn}
Theorem \ref{th:main} suggests
the following correspondence between the boundary vectors 
(\ref{xk}) and (\ref{xb}) with the end shape of the Dynkin diagrams:

\begin{picture}(200,70)(-42,3)

\put(100,50){
\put(-0.4,0){\circle{6}}
\drawline(2,-1.8)(18,-2)
\drawline(2,1.8)(18,2)
\drawline(7,0)(13,-6)
\drawline(7,0)(13,6)
\put(-3,10){\small $0$}
\put(-50,-2){$\langle {\chi}_1(z) |$}
}

\put(100,18){
\put(-0.4,0){\circle{6}}
\drawline(2,-2)(18,-1.8)
\drawline(2,2)(18,1.8)
\drawline(7,6)(13,0)
\drawline(7,-6)(13,0)
\put(-3,10){\small $0$}
\put(-50,-2){$\langle {\chi}_2(z) |$}
}

\put(90,50){
\drawline(82,-2)(98,-1.8)
\drawline(82,2)(98,1.8)
\drawline(93,0)(87,-6)
\drawline(93,0)(87,6)
\put(100.4,0){\circle{6}}
\put(120,-2){$|\chi_1(1)\rangle$}
\put(98,10){\small $n$}
}

\put(90,18){
\drawline(82,-1.8)(98,-2)
\drawline(82,1.8)(98,2)
\drawline(93,6)(87,0)
\drawline(93,-6)(87,0)
\put(100.4,0){\circle{6}}
\put(120,-2){$|\chi_2(1)\rangle$}
\put(98,10){\small $n$}
}

\end{picture}

From this viewpoint it may be natural to interpret 
$S^{2,1}(z)$, which is reducible to $S^{1,2}(z^{1/2}) $ by (\ref{s21}), 
in terms of another $U_q(A^{(2)}_{2n})$ realized as  
the affinization of the classical part $U_q(B_n)$.
(Proposition \ref{pr:repA} corresponds to 
taking the classical part to be $U_q(C_n)$.)
As far as $\langle {\chi}_1(z) |$ and $|\chi_1(1)\rangle$ are 
concerned, the above correspondence agrees 
with the observation made in 
\cite[Remark 7.2]{KS} on the similar result concerning a 
3d $L$ operator.
With regard to $\langle {\chi}_2(z) |$ and $|\chi_2(1)\rangle$,
the relevant affine Lie algebras $A^{(2)}_{2n}$ and 
$C^{(1)}_n$ in this paper are the subalgebras 
of $B^{(1)}_{n+1}$ and $D^{(1)}_{n+2}$
in \cite[Theorem 7.1]{KS}
obtained by folding their Dynkin diagrams. 
\end{remark}

\subsection{Proof}\label{ss:proof}
Let us present an expository proof of 
Theorem \ref{th:main}.
Comparing (\ref{s0}), (\ref{spm})
and (\ref{rnor}), (\ref{Rpmnor}) with 
the gauge transformation (\ref{gtr}) taken into account,
one finds that 
$S^{s,t}(z)$ and ${\tilde R}(z)$ satisfy the 
same normalization condition.
Moreover the conservation law (\ref{claw}) and the commutativity 
(\ref{kR}) are equivalent conditions on the matrices 
acting on ${\hat F}^{\otimes n}\otimes {\hat F}^{\otimes n}$.
Thus it remains to show that 
$S^{s,t}(z)$ satisfies the same equation as 
the gauge transformed version of \eqref{eR} and \eqref{fR} for ${\tilde R}(z)$:
\begin{align}
(\tilde{e}_r\otimes 1+k_r\ot e_r) S^{s,t}(z) &= 
S^{s,t}(z)(1\ot \tilde{e}_r+ e_r\otimes k_r),\label{ce}\\
(1\otimes f_r +\tilde{f}_r \otimes k^{-1}_r) S^{s,t}(z) &= 
S^{s,t}(z)(f_r\otimes 1 + k^{-1}_r\otimes \tilde{f}_r)\label{cf}
\end{align}
$0 \le r \le n$. Here 
$\tilde{e}_r = K^{-1}e_r K,\tilde{f}_r = K^{-1}f_r K$ are the gauge transformed 
Chevalley generators. 
We first treat \eqref{cf}.
The actions of $k^{-1}_r, f_r$ and $\tilde{f}_r$
are to be taken from Proposition \ref{pr:repD},
\ref{pr:repA} and \ref{pr:repC} according to 
$(s,t)=(1,1), (1,2)$ and $(2,2)$, respectively.

Consider the actions of the both sides of (\ref{cf}) on 
a base vector $|{\bf i}\rangle \otimes |{\bf j}\rangle \in 
V_x\otimes V_y$:
\begin{align}
\varrho^{s,t}(z)^{-1}y^{\delta_{r,0}}
(1\otimes f_r +\tilde{f}_r \otimes k^{-1}_r) S^{s,t}(z)
(|{\bf i}\rangle \otimes |{\bf j}\rangle) 
&= 
\sum_{{\bf a}, {\bf b}}A^{\bf{a,b}}_{\bf{i,j}}(z)
|{\bf a}\rangle \otimes |{\bf b}\rangle,\label{Ah}\\
\varrho^{s,t}(z)^{-1}y^{\delta_{r,0}}
S^{s,t}(z)(f_r\otimes 1 + k^{-1}_r\otimes \tilde{f}_r)
(|{\bf i}\rangle \otimes |{\bf j}\rangle) 
&= 
\sum_{{\bf a}, {\bf b}}B^{\bf{a,b}}_{\bf{i,j}}(z)
|{\bf a}\rangle \otimes |{\bf b}\rangle,\label{Bh}
\end{align}
where we have removed the normalization factor 
$\varrho^{s,t}(z)$ (see (\ref{sabij})) 
for simplicity and 
multiplied $y^{\delta_{r,0}}$ to confine the dependence on 
$x$ and $y$ to the ratio $z=x/y$.
We are to show the equality of the matrix elements
$A^{\bf{a,b}}_{\bf{i,j}}(z)=B^{\bf{a,b}}_{\bf{i,j}}(z)$.

For illustration let us consider the case $(s,t)=(1,1)$ and $r=n$.
Then $f_n, \tilde{f}_n$ and $k_n$ are given by Proposition \ref{pr:repD} 
and they only 
touch the $n$ th component in 
$|{\bf i}\rangle$ and $|{\bf j}\rangle$.
The transition of these components in (\ref{Ah}) is traced as follows.
\[
\begin{picture}(120,90)(-45,-80)

\put(0,0){$|i_n\rangle\otimes |j_n\rangle$}

\put(-75,-12){$\Rm^{a_n,b_n-1,c_{n-1}}_{i_n,j_n,c_n}$}
\put(-8,-3){\vector(-1,-1){25}}
\put(-19,-22){$S^{1,1}(z)$}

\put(51,-3){\vector(1,-1){25}}
\put(30,-22){$S^{1,1}(z)$}
\put(63,-12){$\Rm^{a_n-1,b_n,c_{n-1}}_{i_n,j_n,c_n}$}

\put(-65,-40){$|a_n\rangle\otimes |b_n\!-\!1\rangle$}
\put(55,-40){$|a_n\!-\!1\rangle\otimes |b_n\rangle$}

\put(-33,-47){\vector(1,-1){25}}
\put(-19,-57){$1\!\otimes\! f_n$}
\put(-28,-67){$1$}

\put(76,-47){\vector(-1,-1){25}}
\put(28,-57){$\tilde{f}_n\!\otimes\! k^{-1}_n$}
\put(68,-67){$q^{b_n}$}

\put(0,-77){$|a_n\rangle\otimes |b_n\rangle$}
\end{picture}
\]
By this diagram we mean that the substitution of (\ref{sabij}) into 
(\ref{Ah}) yields
\begin{align}
A^{\bf{a,b}}_{\bf{i,j}}(z)
&=\sum_{c_0,\ldots, c_n}
\frac{z^{c_0}}{(q)_{c_n}}X(c_0,\ldots, c_{n-1})
\left(\Rm^{a_n,b_n-1,c_{n-1}}_{i_n,j_n,c_n}
+q^{b_n}\Rm^{a_n-1,b_n,c_{n-1}}_{i_n,j_n,c_n}\right)\nonumber\\
&=\sum_{c_0,\ldots, c_n}
\frac{z^{c_0}}{(q)_{c_n}}X(c_0,\ldots, c_{n-1})
\left((1-q^{c_n})\Rm^{a_n,b_n-1,c_{n-1}}_{i_n,j_n,c_n-1}
+q^{b_n}
\Rm^{a_n-1,b_n,c_{n-1}}_{i_n,j_n,c_n}\right)\label{A2}
\end{align}
for some $X$ which is independent of $z$.
To get the second line
we have just changed 
the dummy summation variable $c_n$ in the first term into $c_n-1$.
This has the effect of letting the two terms have 
the identical constraint  
$b_n+c_{n-1}=j_n+c_n$. See (\ref{cl0}). 
Similarly the diagram for the 
matrix element $B^{\bf{a,b}}_{\bf{i,j}}(z)$ 
(\ref{Bh}) with $r=n$ looks as
\[
\begin{picture}(120,90)(-35,-80)

\put(0,0){$|i_n\rangle\otimes |j_n\rangle$}

\put(-28,-12){$1$}
\put(-8,-3){\vector(-1,-1){25}}
\put(-19,-22){$f_n\!\otimes\! 1$}

\put(51,-3){\vector(1,-1){25}}
\put(27,-22){$k^{-1}_n\!\otimes\!\tilde{f}_n$}
\put(68,-12){$q^{i_n}$}

\put(-65,-40){$|i_n\!+\!1\rangle\otimes |j_n\rangle$}
\put(55,-40){$|i_n\rangle\otimes |j_n\!+\!1\rangle$}

\put(-33,-47){\vector(1,-1){25}}
\put(-19,-57){$S^{1,1}(z)$}
\put(-75,-67){$\Rm^{a_n,b_n,c_{n-1}}_{i_n\!+\!1,j_n,c_n}$}

\put(76,-47){\vector(-1,-1){25}}
\put(30,-57){$S^{1,1}(z)$}
\put(64,-67){$\Rm^{a_n,b_n,c_{n-1}}_{i_n,j_n\!+\!1,c_n}$}

\put(0,-77){$|a_n\rangle\otimes |b_n\rangle$}
\end{picture}
\]
Thus we get 
\begin{align}
B^{\bf{a,b}}_{\bf{i,j}}(z)
&=\sum_{c_0,\ldots, c_n}
\frac{z^{c_0}}{(q)_{c_n}}X(c_0,\ldots, c_{n-1})
\left(\Rm^{a_n,b_n,c_{n-1}}_{i_n\!+\!1,j_n,c_n}
+q^{i_n}\Rm^{a_n,b_n,c_{n-1}}_{i_n,j_n\!+\!1,c_n}\right)\nonumber\\
&=\sum_{c_0,\ldots, c_n}
\frac{z^{c_0}}{(q)_{c_n}}X(c_0,\ldots, c_{n-1})
\left(\Rm^{a_n,b_n,c_{n-1}}_{i_n\!+\!1,j_n,c_n}
+q^{i_n}(1-q^{c_n})
\Rm^{a_n,b_n,c_{n-1}}_{i_n,j_n\!+\!1,c_n-1}\right),\label{B2}
\end{align}
where $X$ is exactly the same function as the one in (\ref{A2}).
In (\ref{A2}) and (\ref{B2}), the $z$-dependence is solely by 
$z^{c_0}$ hence the two $c_0$'s must be identified.
Then from (\ref{sabij}) and $\Rm^{a,b,c}_{i,j,k} \propto \delta^{b+c}_{j+k}$ 
it follows that all the $c_i$'s appearing in (\ref{A2}) and (\ref{B2}) 
are identical.
Therefore the proof of 
$A^{\bf{a,b}}_{\bf{i,j}}(z)= B^{\bf{a,b}}_{\bf{i,j}}(z)$
is reduced to 
\begin{align*}
(1-q^{c_n})\Rm^{a_n,b_n-1,c_{n-1}}_{i_n,j_n,c_n-1}
+q^{b_n}
\Rm^{a_n-1,b_n,c_{n-1}}_{i_n,j_n,c_n}
-\Rm^{a_n,b_n,c_{n-1}}_{i_n\!+\!1,j_n,c_n}
-q^{i_n}(1-q^{c_n})
\Rm^{a_n,b_n,c_{n-1}}_{i_n,j_n\!+\!1,c_n-1}=0.
\end{align*}
But this is just (\ref{app:Li}), which completes the proof 
of (\ref{cf}) for $(s,t)=(1,1)$ and $r=n$.

The essential feature in the above proof is that 
(\ref{cf}) is reduced,  upon substitution of (\ref{sact}), 
to a {\em local} relation in the sequence of $\Rm$'s  
in (\ref{sabij}) with {\em no} sum over $c_0,\ldots, c_n$. 
Another useful fact is that the
actions of $f_r$'s are identical 
in Proposition \ref{pr:repD}, \ref{pr:repA} and \ref{pr:repC} if the vicinity of the  
vertex $r$ of the corresponding Dynkin diagrams has the same shape.
From these considerations, one can attribute the full proof of (\ref{cf}) 
to the following cases.

\begin{enumerate}

\item $r=n$ with $(s,t)=(1,1)$.  
This concerns $D^{(2)}_{n+1}$.
We have just finished the proof.

\item $r=0$ with $(s,t)=(1,1), (1,2)$.
This covers $D^{(2)}_{n+1}$ and $A^{(2)}_{2n}$.

\item $r=n$ with $(s,t)=(1,2), (2,2)$.
This covers $A^{(2)}_{2n}$ and $C^{(1)}_{n}$.

\item $r=0$ with $(s,t)=(2,2)$.
This concerns $C^{(1)}_{n}$.

\item $1 \le r \le n-1$ for any $(s,t)$.
This covers 
$D^{(2)}_{n+1}$, $A^{(2)}_{2n}$ and $C^{(1)}_{n}$.

\end{enumerate}

In what follows, we present the expressions like
(\ref{A2}) and  (\ref{B2})  for each case 
and show how they are identified 
by using the formulas in Appendix \ref{app:3dR}.
In the cases (ii) and (iv),
one needs to cope with the spectral parameter $z=x/y$  
by shifting $c_0$ appropriately.
The case (v) is peculiar in that it requires a proof of a 
{\em quadratic} relation of $\Rm$.

(ii) The $A^{\bf{a,b}}_{\bf{i,j}}(z)$ and $B^{\bf{a,b}}_{\bf{i,j}}(z)$ 
relevant to $(s,t)=(1,1), (1,2)$ are expressed as
\begin{align*}
A^{\bf{a,b}}_{\bf{i,j}}(z)
&=\sum_{c_0,\ldots, c_n}\frac{z^{c_0}(q^2)_{c_0}}{(q)_{c_0}}
\left([b_1+1]\Rm^{a_1,b_1+1,c_0}_{i_1,j_1,c_1}
+q^{-b_1}(1+q^{c_0+1})[a_1+1]\Rm^{a_1+1,b_1,c_0+1}_{i_1,j_1,c_1}\right)
Y_t(c_1,\ldots, c_n),\\
B^{\bf{a,b}}_{\bf{i,j}}(z)
&=\sum_{c_0,\ldots, c_n}
\frac{z^{c_0}(q^2)_{c_0}}{(q)_{c_0}}
\left((1+q^{c_0+1})[i_1]\Rm^{a_1,b_1,c_0+1}_{i_1-1,j_1,c_1}
+q^{-i_1}[j_1]\Rm^{a_1,b_1,c_0}_{i_1,j_1-1,c_1}\right)
Y_t(c_1,\ldots, c_n)
\end{align*}
for some $Y_t$ which is independent of $z$.
These expressions are identified by (\ref{app:Ask}).

(iii) The $A^{\bf{a,b}}_{\bf{i,j}}(z)$ and $B^{\bf{a,b}}_{\bf{i,j}}(z)$ 
relevant to $(s,t)=(1,2), (2,2)$ are expressed as
\begin{align*}
A^{\bf{a,b}}_{\bf{i,j}}(z)
&=\sum_{c_0,\ldots, c_n}
\frac{z^{c_0}}{(q^4)_{c_n}}Z_s(c_0,\ldots, c_{n-1})
\left((1-q^{4c_n})\Rm^{a_n,b_n-2,c_{n-1}}_{i_n,j_n,2c_n-2}
+q^{2b_n}\Rm^{a_n-2,b_n,c_{n-1}}_{i_n,j_n,2c_n}\right),\\
B^{\bf{a,b}}_{\bf{i,j}}(z)
&=\sum_{c_0,\ldots, c_n}
\frac{z^{c_0}}{(q^4)_{c_n}}Z_s(c_0,\ldots, c_{n-1})
\left(\Rm^{a_n,b_n,c_{n-1}}_{i_n\!+\!2,j_n,2c_n}
+q^{2i_n}(1-q^{4c_n})\Rm^{a_n,b_n,c_{n-1}}_{i_n,j_n\!+\!2,2c_n-2}\right)
\end{align*}
for some $Z_s$ which is independent of $z$.
These expressions are identified by (\ref{app:Ngm}).

(iv) The $A^{\bf{a,b}}_{\bf{i,j}}(z)$ and $B^{\bf{a,b}}_{\bf{i,j}}(z)$ 
relevant to $(s,t)=(2,2)$ are expressed as
\begin{align*}
A^{\bf{a,b}}_{\bf{i,j}}(z)
&=\sum_{c_0,\ldots, c_n}
\frac{z^{c_0}(q^2)_{2c_0}}{(q^4)_{c_0}}
\Bigl([b_1+2][b_1+1]\Rm^{a_1,b_1+2,2c_0}_{i_1,j_1,c_1}\\
&\qquad \qquad +q^{-2b_1}(1-q^{4c_0+2})[a_1+2][a_1+1]
\Rm^{a_1+2,b_1,2c_0+2}_{i_1,j_1,c_1}
\Bigr)W(c_1,\ldots, c_n),\\
B^{\bf{a,b}}_{\bf{i,j}}(z)
&=\sum_{c_0,\ldots, c_n}
\frac{z^{c_0}(q^2)_{2c_0}}{(q^4)_{c_0}}
\Bigl((1-q^{4c_0+2})[i_1][i_1-1]\Rm^{a_1,b_1,2c_0+2}_{i_1-2,j_1,c_1}\\
&\qquad \qquad 
+q^{-2i_1}[j_1][j_1-1]\Rm^{a_1,b_1,2c_0}_{i_1,j_1-2,c_1}
\Bigr)W(c_1,\ldots, c_n)
\end{align*}
for some $W$ which is independent of $z$.
These expressions are identified by (\ref{app:Ymi}).

(v) The $f_r$ and $k_r$ with $1 \le r \le n-1$ 
concern the $r$ th and the $(r\!+\!1)$ th components of $F^{\otimes n}$ only.
The diagram for (\ref{Ah})  tracing them looks as
\[
\begin{picture}(120,90)(-45,-80)

\put(-23,4){$|i_r, i_{r+1}\rangle\otimes |j_r, j_{r+1}\rangle$}

\put(-8,-3){\vector(-1,-1){25}}
\put(-19,-22){$S^{s,t}(z)$}

\put(51,-3){\vector(1,-1){25}}
\put(32,-22){$S^{s,t}(z)$}

\put(-125,-40){$|a_r,a_{r+1}\rangle\otimes |b_r\!-\!1,b_{r+1}\!+\!1\rangle$}
\put(60,-40){$|a_r\!-\!1,a_{r+1}\!+\!1\rangle\otimes |b_r,b_{r+1}\rangle$}

\put(-33,-47){\vector(1,-1){25}}
\put(-19,-57){$1\!\otimes\! f_r$}
\put(-62,-67){$[b_{r+1}\!+\!1]$}

\put(76,-47){\vector(-1,-1){25}}
\put(28,-57){$\tilde{f}_r\!\otimes\! k^{-1}_r$}
\put(68,-67){$q^{b_r-b_{r+1}}[a_{r+1}\!+\!1]$}

\put(-23,-85){$|a_r,a_{r+1}\rangle\otimes |b_r,b_{r+1}\rangle$}
\end{picture}
\]
Thus we have
\begin{equation}\label{A3}
\begin{split}
&A^{\bf{a,b}}_{\bf{i,j}}(z)
=\sum_{c_0,\ldots, c_n}z^{c_0}U_{s,t}(c_0, \ldots, c_{r-1}, c_{r+2}, \ldots, c_n)\\
&\quad\times\Bigl([b_{r+1}\!+\!1]\Rm^{a_r, b_r-1,c_{r-1}}_{i_r,j_r,c_r-1}
\Rm^{a_{r+1},b_{r+1}+1,c_r-1}_{i_{r+1},j_{r+1},c_{r+1}}
+q^{b_r-b_{r+1}}[a_{r+1}\!+\!1]
\Rm^{a_r-1,b_r,c_{r-1}}_{i_r,j_r,c_r}
\Rm^{a_{r+1}+1,b_{r+1},c_r}_{i_{r+1},j_{r+1},c_{r+1}}\Bigr)
\end{split}
\end{equation}
for some $U_{s,t}$ which is independent of $z$.
We have shifted $c_r$ to $c_r-1$  in the first term
by the reason similar to (\ref{A2}) and (\ref{B2}). 
Similarly the diagram for (\ref{Bh}) looks as
\[
\begin{picture}(120,90)(-45,-80)

\put(-23,4){$|i_r, i_{r+1}\rangle\otimes |j_r, j_{r+1}\rangle$}

\put(-8,-3){\vector(-1,-1){25}}
\put(-19,-22){$f_r\!\otimes\! 1$}

\put(-50,-12){$[i_{r+1}]$}

\put(51,-3){\vector(1,-1){25}}
\put(26,-22){$k^{-1}_r\!\otimes\! \tilde{f}_r$}

\put(68,-12){$q^{i_r-i_{r+1}}[j_{r+1}]$}

\put(-125,-40){$|i_r\!+\!1,i_{r+1}\!-\!1\rangle\otimes |j_r,j_{r+1}\rangle$}
\put(60,-40){$|i_r,i_{r+1}\rangle\otimes |j_r\!+\!1,j_{r+1}\!-\!1\rangle$}

\put(-33,-47){\vector(1,-1){25}}
\put(-19,-57){$S^{s,t}(z)$}

\put(76,-47){\vector(-1,-1){25}}
\put(32,-57){$S^{s,t}(z)$}

\put(-23,-85){$|a_r,a_{r+1}\rangle\otimes |b_r,b_{r+1}\rangle$}
\end{picture}
\]
This leads to the expression
\begin{equation}\label{B3}
\begin{split}
&B^{\bf{a,b}}_{\bf{i,j}}(z)
=\sum_{c_0,\ldots, c_n}z^{c_0}U_{s,t}(c_0, \ldots, c_{r-1}, c_{r+2}, \ldots, c_n)\\
&\quad\times\Bigl([i_{r+1}]\Rm^{a_r, b_r,c_{r-1}}_{i_r+1,j_r,c_r}
\Rm^{a_{r+1},b_{r+1},c_r}_{i_{r+1}-1,j_{r+1},c_{r+1}}
+q^{i_r-i_{r+1}}[j_{r+1}]
\Rm^{a_r,b_r,c_{r-1}}_{i_r,j_r+1,c_r-1}
\Rm^{a_{r+1},b_{r+1},c_r-1}_{i_{r+1},j_{r+1}-1,c_{r+1}}\Bigr)
\end{split}
\end{equation}
with the same $U_{s,t}$ as (\ref{A3}).
This time the shift of $c_r$ to $c_r-1$ has been done in the second term.
Now that all the $c_i$'s can be identified in (\ref{A3}) and (\ref{B3}),
their equality follows from (\ref{app:air}).
This completes the proof of (\ref{cf}).
The relation (\ref{ce}) can be verified similarly by using 
(\ref{rtb}), (\ref{hmk}), (\ref{hnt}) and (\ref{szk}).
\qed

\appendix
\section{Brief guide to 3d $R$}\label{app:3dR}
\subsection{Origin in quantized coordinate ring}
Let us summarize the basic facts on the 3d $R$ $\Rm$ 
that has played a central role in the paper
from the viewpoint of the quantized coordinate ring $A_q(sl_3)$
following \cite{KV,KO1}.
See also \cite{BMS,BS,KOY,S} for more aspects.
The $A_q(sl_3)$ is 
a Hopf algebra generated by $T=(t_{ij})_{1 \le i,j \le 3}$ satisfying the relations
\begin{equation*}
\begin{split}
&[t_{ik}, t_{jl}]=\begin{cases}0 &(i<j, k>l),\\
(q-q^{-1})t_{jk}t_{il} & (i<j, k<l),
\end{cases}\\
&t_{ik}t_{jk} = q t_{jk}t_{ik}\; (i<j),\quad 
t_{ki}t_{kj} = q t_{kj}t_{ki}\; (i<j).
\end{split}
\end{equation*} 
The coproduct is given by
$\Delta(t_{ij}) = \sum_k t_{ik}\otimes t_{kj}$.
Let $F = \bigoplus_{m\ge 0}\Q(q)|m\rangle$ be the Fock space as in the main text.
The $A_q(sl_3)$ has irreducible representations
$\pi_i: A_q(sl_3) \rightarrow \mathrm{End}(F)\; (i=1,2)$ as
\begin{align*}
&\pi_1(T) =
\begin{pmatrix}
\mu_1{\rm {\bf a}}^- &  \alpha_1{\rm {\bf k}} & 0 \\
-q\alpha^{-1}_1{\rm {\bf k}} & \mu^{-1}_1{\rm {\bf a}}^+ & 0\\
0 & 0 & 1
\end{pmatrix},
\quad
\pi_2(T) =
\begin{pmatrix}
1 & 0 & 0\\
0 & \mu_2{\rm {\bf a}}^- &  \alpha_2{\rm {\bf k}}\\
0 & -q\alpha^{-1}_2{\rm {\bf k}} & \mu^{-1}_2{\rm {\bf a}}^+ 
\end{pmatrix},\\
&{\rm{\bf k}}|m\rangle = q^m |m\rangle,\;
{\rm{\bf a}}^+|m\rangle = |m+1\rangle,\;
{\rm{\bf a}}^-|m\rangle = (1-q^{2m})|m-1\rangle.
\end{align*}
The parameters $\mu_i, \alpha_i$ 
are set to be 1 in the sequel as they do not influence the 
construction in a nontrivial way.
The $\pi_1$ and $\pi_2$ are called fundamental representations. 
Let $\pi_{121}$ and $\pi_{212}$ be their  
tensor product representations on $F^{\otimes 3}$ obtained by evaluating 
the coproduct $\Delta(t_{ij}) = \sum_{k,l}t_{ik}\otimes t_{kl}\otimes t_{lj}$
by $\pi_1\otimes \pi_2\otimes \pi_1$ and 
$\pi_2\otimes \pi_1\otimes \pi_2$, respectively.
It is known that they are both irreducible and equivalent.
Thus there is the unique map 
$\Phi \in \mathrm{End}(F^{\otimes 3})$ 
satisfying the intertwining relation 
$\Phi \circ \pi_{121} = \pi_{212} \circ \Phi$ up to normalization.
Let $\sigma \in \mathrm{End}(F^{\otimes 3})$ be the reversal of the 
components 
$\sigma(u\otimes v \otimes w) = w\otimes v \otimes u$.
Then the 3d $R$ is identified as $\Rm = \Phi \circ \sigma$ with the 
normalization 
$\Rm (|0\rangle \otimes|0\rangle \otimes|0\rangle)  
= |0\rangle \otimes|0\rangle \otimes|0\rangle$.
In short 3d $R$ $\Rm$ is the intertwiner of $A_q(sl_3)$.
The tetrahedron equation (\ref{TE}) is a corollary of the fact that
the similar intertwiner for $A_q(sl_4)$ can be constructed as 
quartic products of $\Rm$ in two different forms \cite{KV}.
By investigating the intertwining relation
one can show (cf. \cite[Proposition 2.4]{KO1})
\begin{align}\label{rtb}
\Rm^{a,b,c}_{i,j,k}= \Rm^{c,b,a}_{k,j,i},\quad
\Rm^{a,b,c}_{i,j,k}=
\frac{(q^2)_i(q^2)_j(q^2)_k}{(q^2)_a(q^2)_b(q^2)_c}
\Rm_{a,b,c}^{i,j,k}.
\end{align}

\subsection{Intertwining relations of 3d $\boldsymbol{R}$}
By the definition the 3d $R$ satisfies the intertwining relation
$\Rm \circ \pi'_{121} = \pi_{212} \circ \Rm$, where
$\pi'_{121}  = \sigma \circ \pi_{121} \circ \sigma$.
Evaluating the generator $t_{rs}$ on the both sides and picking the 
matrix elements for 
$|i\rangle \otimes |j\rangle \otimes |k\rangle  
\mapsto
|a\rangle \otimes |b\rangle \otimes |c\rangle$
by using (\ref{Rabc}) leads to useful recursion relations on 
$\Rm^{a,b,c}_{i,j,k}$:
\begin{align}
t_{11}:\;&q^{i+k+1} \left(1\!-\!q^{2 j}\right) \Rm^{a,b,c}_{i,j-1,k}
-\left(1\!-\!q^{2 i}\right) \left(1\!-\!q^{2 k}\right)\Rm^{a,b,c}_{i-1,j,k-1}
+\left(1\!-\!q^{2 b+2}\right) \Rm^{a,b+1,c}_{i,j,k}=0,
\nonumber\\
t_{12}:\;&q^k\left(1\!-\!q^{2 j}\right) \Rm^{a,b,c}_{i+1,j-1,k}
+q^i \left(1\!-\!q^{2 k}\right) \Rm^{a,b,c}_{i,j,k-1}
-q^b \left(1\!-\!q^{2 c+2}\right)\Rm^{a,b,c+1}_{i,j,k}=0,
\nonumber\\
t_{21}:\;&q^{i} \left(1\!-\!q^{2 j}\right)\Rm^{a,b,c}_{i,j-1,k+1}
+q^{k}\left(1\!-\!q^{2 i}\right) \Rm^{a,b,c}_{i-1,j,k}
-q^{b}\left(1\!-\!q^{2a+2}\right) \Rm^{a+1,b,c}_{i,j,k}=0,
\nonumber\\ 
t_{22}:\;& q \left(q^{a+c}\!-\!q^{i+k}\right) \Rm^{a,b,c}_{i,j,k}
+\left(1\!-\!q^{2j}\right) \Rm^{a,b,c}_{i+1,j-1,k+1}
-\left(1\!-\!q^{2 a+2}\right) \left(1\!-\!q^{2 c+2}\right)
\Rm^{a+1,b-1,c+1}_{i,j,k}=0,
\nonumber\\
t_{23}:\;&q^j \Rm^{a,b,c}_{i,j,k+1}-q^a\Rm^{a,b,c-1}_{i,j,k}
-q^c \left(1\!-\!q^{2a+2}\right) \Rm^{a+1,b-1,c}_{i,j,k}=0,
\nonumber\\
t_{32}:\;&q^{c}\Rm^{a-1,b,c}_{i,j,k}
-q^{j} \Rm^{a,b,c}_{i+1,j,k}
+q^{a}\left(1-q^{2 c+2}\right)\Rm^{a,b-1,c+1}_{i,j,k}=0,
\label{t32}\\
t_{33}:\;&q^{a+c+1}
   \Rm^{a,b-1,c}_{i,j,k}-\Rm^{a-1,b,c-1}_{i,j,k}+\Rm^{a,b,c}_{i,j+1,k}=0.
\label{t33}
\end{align} 
We have skipped 
$t_{13}$ and $t_{31}$ as they just give 
$(q^{i+j}-q^{a+b})\Rm^{a,b,c}_{i,j,k} = 
(q^{j+k}-q^{b+c})\Rm^{a,b,c}_{i,j,k}=0$,
which is the origin of the conservation law (\ref{cl0}).
The formula (\ref{Rex}) was derived by solving 
(\ref{t32}) and (\ref{t33}) \cite{KO1}. 
The relation for $t_{rs}$ here is transformed into the one for 
$t_{4-s,4-r}$ by using the latter property in (\ref{rtb}).
 
It is known that $\Rm=\Rm^{-1}$ hence 
$\pi'_{121}\circ \Rm = \Rm \circ \pi_{212}$ also holds.
The recursion relations corresponding to this read
\begin{align}
t_{11}:\;&q^{a+c+1}\left(1\!-\!q^{2 b+2}\right) \Rm^{a,b+1,c}_{i,j,k}
+\left(1\!-\!q^{2 j}\right)\Rm^{a,b,c}_{i,j-1,k}
-\left(1\!-\!q^{2 a+2}\right) \left(1\!-\!q^{2c+2}\right)
\Rm^{a+1,b,c+1}_{i,j,k}=0,\label{ti11}\\
t_{12}:\;&q^j \left(1\!-\!q^{2 k}\right)\Rm^{a,b,c}_{i,j,k-1}
- q^c\left(1\!-\!q^{2 b+2}\right)\Rm^{a-1,b+1,c}_{i,j,k}
-q^a \left(1\!-\!q^{2 c+2}\right)\Rm^{a,b,c+1}_{i,j,k}=0,
\nonumber\\
t_{21}:\;&q^{j}\left(1\!-\!q^{2 i}\right) \Rm^{a,b,c}_{i-1,j,k}
-q^{a}\left(1\!-\!q^{2b+2}\right)\Rm^{a,b+1,c-1}_{i,j,k}
-q^{c}\left(1\!-\!q^{2a+2}\right)\Rm^{a+1,b,c}_{i,j,k}=0,
\label{ti21}\\
t_{22}:\;&\left(1\!-\!q^{2 b+2}\right)\Rm^{a-1,b+1,c-1}_{i,j,k}
-\left(1\!-\!q^{2 i}\right) \left(1\!-\!q^{2k}\right) \Rm^{a,b,c}_{i-1,j+1,k-1}
-q\left(q^{a+c}\!-\!q^{i+k}\right) \Rm^{a,b,c}_{i,j,k}=0,
\nonumber\\
t_{23}:\;&q^i\Rm^{a,b,c}_{i,j,k+1}
-q^b\Rm^{a,b,c-1}_{i,j,k}
+ q^k\left(1\!-\!q^{2i}\right)\Rm^{a,b,c}_{i-1,j+1,k}=0,
\nonumber\\
t_{32}:\;&q^{b} \Rm^{a-1,b,c}_{i,j,k}
-q^{i}\left(1\!-\!q^{2 k}\right)\Rm^{a,b,c}_{i,j+1,k-1}
-q^{k}\Rm^{a,b,c}_{i+1,j,k}=0,
\label{ti32}\\
t_{33}:\;&q^{i+k+1}\Rm^{a,b,c}_{i,j+1,k}
+\Rm^{a,b-1,c}_{i,j,k}
-\Rm^{a,b,c}_{i+1,j,k+1}=0.
\label{ti33}
\end{align}
Again we have omitted $t_{13}$ and $t_{31}$ leading to 
the conservation law.
The relation for $t_{rs}$ here is transformed into the one for 
$t_{4-s,4-r}$ by combining the two properties in (\ref{rtb}).

Although not all of the above recursion relations are necessary 
in this paper, we have listed them for convenience 
in possible future works.

\subsection{Lemma}
Now we collect the relations necessary in the proof of Theorem \ref{th:main}.
We recall that $[m]=[m]_q$ is defined in the end of Section \ref{sec:1}.
\begin{lemma}\label{le:sor}
The following relations hold:
\begin{align}
&(1-q^k)\Rm^{a,b-1,c}_{i,j,k-1}
+q^b\Rm^{a-1,b,c}_{i,j,k}-\Rm^{a,b,c}_{i+1,j,k}
-q^i(1-q^k)\Rm^{a,b,c}_{i,j+1,k-1}=0,\label{app:Li}\\
&[b+1]\Rm^{a,b+1,c}_{i,j,k}
+q^{-b}(1+q^{c+1})[a+1]\Rm^{a+1,b,c+1}_{i,j,k}
-(1+q^{c+1})[i]\Rm^{a,b,c+1}_{i-1,j,k}
-q^{-i}[j]\Rm^{a,b,c}_{i,j-1,k}=0,\label{app:Ask}\\
&(1-q^{2k})\Rm^{a,b-2,c}_{i,j,k-2}
+q^{2b}\Rm^{a-2,b,c}_{i,j,k}
-\Rm^{a,b,c}_{i+2,j,k}
-q^{2i}(1-q^{2k})\Rm^{a,b,c}_{i,j+2,k-2}=0,
\label{app:Ngm}\\
&[b+2][b+1]\Rm^{a,b+2,c}_{i,j,k}
+q^{-2b}(1-q^{2c+2})[a+2][a+1]\Rm^{a+2,b,c+2}_{i,j,k}\nonumber\\
&-(1-q^{2c+2})[i][i-1]\Rm^{a,b,c+2}_{i-2,j,k}
-q^{-2i}[j][j-1]\Rm^{a,b,c}_{i,j-2,k}=0.
\label{app:Ymi}
\end{align}
\end{lemma}
\begin{proof}
Regard (\ref{ti32}) as a recursion increasing $a$ by one keeping 
$b$ and $c$.
Similarly (\ref{ti33}) provides a recursion increasing $b$ by one 
keeping $a$ and $c$.
One can apply them to the first two terms in (\ref{app:Li})
to bring all the terms into the form 
$\Rm^{a,b,c}_{\bullet, \bullet, \bullet}$.
The result turns out to be identically zero.
The equality (\ref{app:Ngm}) is shown in the same way 
by applying each recursion {\em twice} to the first two terms therein. 
Finally (\ref{app:Ask}) and (\ref{app:Ymi})
are derived by applying the latter relation in (\ref{rtb})
to (\ref{app:Li}) and (\ref{app:Ngm}), respectively.
\end{proof}

\begin{lemma}\label{le:air}
The following quadratic relation among $\Rm$ holds:
\begin{equation}\label{app:air}
\begin{split}
&[b'\!+\!1]\Rm^{a, b-1,c}_{i,j,k-1}
\Rm^{a',b'+1,k-1}_{i',j',k'}
+q^{b-b'}[a'\!+\!1]
\Rm^{a-1,b,c}_{i,j,k}
\Rm^{a'+1,b',k}_{i',j',k'}\\
&-[i']\Rm^{a, b,c}_{i+1,j,k}
\Rm^{a',b',k}_{i'-1,j',k'}
-q^{i-i'}[j']
\Rm^{a,b,c}_{i,j+1,k-1}
\Rm^{a',b',k-1}_{i',j'-1,k'}=0.
\end{split}
\end{equation}
\end{lemma}
\begin{proof}
From  (\ref{ti33}) and (\ref{ti11}) one has
\begin{align*}
&\Rm^{a,b-1,c}_{i,j,k-1}=
\Rm^{a,b,c}_{i+1,j,k}-q^{i+k}\Rm^{a,b,c}_{i,j+1,k-1},\\
&[j']\Rm^{a',b',k-1}_{i',j'-1,k'}
=q^{a'-j'+1}[a'+1](1-q^{2k})\Rm^{a'+1,b',k}_{i',j',k'}
-q^{i'+k}[b'+1]\Rm^{a',b'+1,k-1}_{i',j',k'}.
\end{align*}
By substituting them to the first and the last $\Rm$ in (\ref{app:air}),
the LHS becomes
\begin{align*}
&[b'\!+\!1]\bigl(\Rm^{a,b,c}_{i+1,j,k}-q^{i+k}
\underline{\Rm^{a,b,c}_{i,j+1,k-1}}\bigr)
\Rm^{a',b'+1,k-1}_{i',j',k'}
+q^{b-b'}[a'\!+\!1]
\Rm^{a-1,b,c}_{i,j,k}
\Rm^{a'+1,b',k}_{i',j',k'}\\
&-[i']\Rm^{a, b,c}_{i+1,j,k}
\Rm^{a',b',k}_{i'-1,j',k'}
-q^{i-i'}\Rm^{a,b,c}_{i,j+1,k-1}
\bigl(q^{a'-j'+1}[a'\!+\!1](1\!-\!q^{2k})\Rm^{a'+1,b',k}_{i',j',k'}
-q^{i'+k}[b'\!+\!1]\underline{\Rm^{a',b'+1,k-1}_{i',j',k'}}\bigr).
\end{align*}
The contributions from the underlined terms cancel.
The remaining four terms are grouped as
\begin{equation}\label{app:rei}
\begin{split}
&\Rm^{a,b,c}_{i+1,j,k}\bigl(
[b'+1]\Rm^{a',b'+1,k-1}_{i',j',k'}-[i']\Rm^{a',b',k}_{i'-1,j',k'}\bigr)\\
&+q^{-b'}[a'+1]
\bigl(q^b\Rm^{a-1,b,c}_{i,j,k}
-q^{i+\phi}(1-q^{2k})\Rm^{a,b,c}_{i,j+1,k-1}\bigr)
\Rm^{a'+1,b',k}_{i',j',k'}
\end{split}
\end{equation}
with $\phi=a'+b'-i'-j'+1$ which is zero due the conservation law 
for $\Rm^{a'+1,b',k}_{i',j',k'}$.
The combination in the first parenthesis in (\ref{app:rei})
is equal to $-q^{k-b'}[a'+1]\Rm^{a'+1,b',k}_{i',j',k'}$
due to (\ref{ti21}) and $\phi=0$.
The one in the second parenthesis
is equal to $q^{k}\Rm^{a,b,c}_{i+1,j,k}$
by (\ref{ti32}).
Thus (\ref{app:rei}) vanishes.
\end{proof}

We note that $\Rm$ satisfies 
further relations
\begin{align}
&q^{a+1}(1+q^c)\Rm^{a,b-1,c}_{i,j,k}-
\Rm^{a-1,b,c-1}_{i,j,k}-q^{j+1}\Rm^{a,b,c-1}_{i+1,j,k}
+(1+q^c)\Rm^{a,b,c}_{i,j+1,k}=0,\label{hmk}\\
&\Rm^{a-2,b,c}_{i,j,k}+q^{2a+2}(1-q^{2c+2})\Rm^{a,b-2,c+2}_{i,j,k}
-(1-q^{2c+2})\Rm^{a,b,c+2}_{i,j+2,k}
-q^{2j+2}\Rm^{a,b,c}_{i+2,j,k}=0,\label{hnt}\\
\begin{split}
&[a+1]\Rm^{a+1,b,c}_{i,j,k}
\Rm^{a'-1,b',k}_{i',j',k'}+
q^{-a+a'}[b+1]\Rm^{a,b+1,c}_{i,j,k+1}
\Rm^{a',b'-1,k+1}_{i',j',k'}\\
&-
[j]\Rm^{a,b,c}_{i,j-1,k+1}
\Rm^{a',b',k+1}_{i',j'+1,k'}-
q^{-j+j'}[i]\Rm^{a,b,c}_{i-1,j,k}
\Rm^{a',b',k}_{i'+1,j',k'}=0.
\end{split}\label{szk}
\end{align}
They are proved similarly to Lemma \ref{le:sor} and Lemma \ref{le:air}.

\section{Trace reduction of 
tetrahedron equation and $q$-oscillator representation of $U_q(A^{(1)}_{n-1})$}
\label{app:A}

For comparison
we include an exposition of type $A$ case
which is known to be related to the trace reduction of the tetrahedron equation
to the Yang-Baxter equation \cite{BS}.

Let ${\bf h}$ be the operator on $F$ acting as 
${\bf h}|m\rangle = m|m\rangle$.
Then the conservation law (\ref{cl0}) implies the commutativity
$[\Rm_{1,2,3}, x^{{\bf h}_1}(xy)^{{\bf h}_2} y^{{\bf h}_3}]=0$.
Multiplying  
$\Rm^{-1}_{4,5,6}x^{{\bf h}_4}(xy)^{{\bf h}_5} y^{{\bf h}_6}$
from the left to (\ref{TEn}) and taking the trace over 
$\overset{4}{F}\otimes\overset{5}{F}\otimes\overset{6}{F}$, one finds that 
$S^{\mathrm{tr}}_{\boldsymbol{\alpha, \beta}}(z) 
=\mathrm{Tr}_3(z^{{\bf h}_3}
\Rm_{\alpha_1, \beta_1, 3}
\Rm_{\alpha_2, \beta_2, 3}\cdots
\Rm_{\alpha_n, \beta_n, 3})
\in \mathrm{End}
(\overset{\boldsymbol\alpha}{F}\otimes \overset{\boldsymbol\beta}{F})$
satisfies the Yang-Baxter equation (\ref{sybe}).
The matrix elements are given by
\begin{align*}
&S^{\mathrm{tr}}(z)\bigl(|{\bf i}\rangle \otimes |{\bf j}\rangle\bigr)
= \sum_{{\bf a},{\bf b}}
S^{\mathrm{tr}}(z)^{{\bf a},{\bf b}}_{{\bf i},{\bf j}}
|{\bf a}\rangle \otimes |{\bf b}\rangle,\\
&S^{\mathrm{tr}}(z)^{{\bf a},{\bf b}}_{{\bf i},{\bf j}}
=\sum_{c_0, \ldots, c_{n-1}\ge 0} 
z^{c_0}
\Rm^{a_1, b_1, c_0}_{i_1, j_1, c_1}
\Rm^{a_2, b_2, c_1}_{i_2, j_2, c_2}\cdots
\Rm^{a_{n\!-\!1}, b_{n\!-\!1}, c_{n\!-\!2}}_{i_{n\!-\!1}, j_{n\!-\!1}, c_{n\!-\!1}}
\Rm^{a_n, b_n, c_{n\!-\!1}}_{i_n, j_n, c_0},
\end{align*}
which is the trace version of the formula (\ref{sact})--(\ref{sabij}).
For instance one has
\begin{align}
&S^{\mathrm{tr}}(z)^{{\bf a},{\bf 0}}_{{\bf a},{\bf 0}}
=\frac{1}{1-zq^{|{\bf a}|}},\nonumber\\
&S^{\mathrm{tr}}(z)^{m{\bf e}_k, l{\bf e}_k}_{m{\bf e}_k, l{\bf e}_k}
=(-q)^{l-m}
S^{\mathrm{tr}}(z)^{l{\bf e}_k, m{\bf e}_k}_{l{\bf e}_k, m{\bf e}_k}
=z^l\frac{(q^{m-l+2}z^{-1};q^2)_l}
{(q^{m-l}z;q^2)_{l+1}}\label{slm}
\end{align}
for any $k$.
See Section \ref{ss:ex} for the notation.
The conservation law takes the form 
\begin{align}\label{Aclaw}
S^{\mathrm{tr}}(z)^{{\bf a},{\bf b}}_{{\bf i},{\bf j}} = 0\;\;
\text{unless}\;\; |{\bf a}| = |{\bf i}|,\; |{\bf b}|  = |{\bf j}|,
\end{align}
therefore $S^{\mathrm{tr}}(z)$ splits into infinitely many 
irreducible components.

The $S^{\mathrm{tr}}(z)$ 
stems from the $q$-oscillator representation of 
$U_q(A^{(1)}_{n-1})\, (n\ge 2)$, which we shall now explain.
The algebra $U_q(A^{(1)}_{n-1})$ is defined 
by (\ref{uqdef}) with $n$ replaced by $n-1$, 
$a_{ij}=2\delta_{i,j}-\delta_{|i-j|,1}-\delta_{|i-j|,n}$ 
and $q_i=q$ for $0 \le i \le n-1$.
For $n\ge 3$ the Dynkin diagram has circle shape:
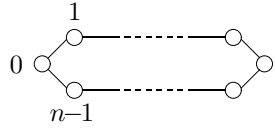
\[
\begin{picture}(126,43)(0,-20)

\multiput(0,10)(0,-20){2}{
\put(40,0){\circle{6}}
\put(100,0){\circle{6}}
\put(43,0){\line(1,0){14}}
\put(83,0){\line(1,0){14}}
\multiput(59,0)(4,0){6}{\line(1,0){2}} 
}

\put(27.8,0){\circle{6}}
\put(30.1,2.2){\line(1,1){7}}
\put(30.1,-2.2){\line(1,-1){7}}

\put(111.8,0){\circle{6}}
\put(110.1,2.2){\line(-1,1){7}}
\put(110.1,-2.2){\line(-1,-1){7}}

\put(18,3){\makebox(0,0)[t]{$0$}}
\put(40, 23){\makebox(0,0)[t]{$1$}}

\put(39, -15){\makebox(0,0)[t]{$n\!\! -\!\! 1$}}

\end{picture}
\]
It is easy to see that the action of the generators  
\begin{equation}\label{haya}
\begin{split}
e_j|{\bf m}\rangle 
&= x^{\delta_{j,0}}[m_j]|{\bf m}-{\bf e}_j+{\bf e}_{j+1}\rangle,\\
f_j|{\bf m}\rangle &= 
x^{-\delta_{j,0}}[m_{j+1}]|{\bf m}+{\bf e}_j-{\bf e}_{j+1}\rangle,\\
k_j|{\bf m}\rangle &= q^{-m_{j}+m_{j+1}}|{\bf m}\rangle
\end{split}
\end{equation}
defines a $U_q(A^{(1)}_{n-1})$ module structure on 
$F^{\otimes n}$.
Here the indices are to be understood mod $n$
and $x$ is a nonzero parameter.
The representation (\ref{haya}) essentially goes back to \cite{Ha}.

Denote the representation space by $V_x=F^{\otimes n}[x,x^{-1}]$.
It decomposes as
\begin{align*}
V_x = \bigoplus_{l\ge 0} V_{x,l},\qquad
V_{x,l} = \bigoplus_{{\bf m}\in (\Z_{\ge 0})^n,
\; |{\bf m}|=l}\Q(q)|{\bf m}\rangle,
\end{align*}
where the symbol $ |{\bf m}|$ is defined under (\ref{d22}).
The component $V_{x,l}$ is isomorphic, as a module over the 
classical subalgebra $U_q(A_{n-1}) 
= \langle e_i,f_i, k^{\pm 1}_i \rangle_{1 \le i <n}$, to 
the highest weight representation with  
highest weight $l\varpi_{n-1}$\footnote{$\varpi_j$ denotes the $j$ th 
fundamental weight.
By a conventional reason the highest weight here 
is a dual of  $l\varpi_1$ in \cite{BS}.
}.
Its highest weight vector is $|l{\bf e}_n\rangle$.
Let 
$R = R_{l,m}(z) \in 
\mathrm{End}(V_{x,m} \otimes V_{y,l})$  $(z=x/y)$
be the quantum $R$ matrix.
Namely $R$ satisfies (\ref{drrd}) for $U_q=U_q(A^{(1)}_{n-1})$.
We normalize it by
$R_{m,l}(z)(|m{\bf e}_n\rangle \otimes |l{\bf e}_n\rangle)
=z^l\frac{(q^{m-l+2}z^{-1};q^2)_l}
{(q^{m-l}z;q^2)_{l+1}}
|m{\bf e}_n\rangle \otimes |l{\bf e}_n\rangle$.
Up to the normalization of $R$ matrices and 
conventional difference, the following equality was announced in \cite{BS}.
\begin{proposition}
\begin{align*}
S^{\mathrm{tr}}(z) = \bigoplus_{m,l \ge 0} R_{m,l}(z).
\end{align*}
\end{proposition}
\begin{proof}
The conservation law (\ref{Aclaw}) tells that 
$S^{\mathrm{tr}}(z)$ satisfies (\ref{kR})
and splits in the same pattern as the RHS.
By (\ref{slm}) the both sides have the same normalization on 
$|m{\bf e}_n\rangle \otimes |l{\bf e}_n\rangle$.
Thus it suffices to show 
$(1\otimes f_r +f_r \otimes k^{-1}_r) S^{\mathrm{tr}}(z) = 
S^{\mathrm{tr}}(z) (f_r\otimes 1 + k^{-1}_r\otimes f_r)$
for $0 \le r \le n-1$.
As the case (v) in the proof of Theorem \ref{th:main}, 
this reduces exactly to Lemma \ref{app:air} including $r=0$ case.
\end{proof}

\section*{Acknowledgments}
The authors thank Takahiro Hayashi for informing the references 
\cite{MMNNSU,P}.
This work is supported by Grants-in-Aid for
Scientific Research No.~23340007, No.~24540203 and
No.~23654007 from JSPS.

\end{document}